\pdfoutput=1
\documentclass[floatfix,pre,twocolumn,english]{revtex4}
\usepackage{graphicx, subfigure}
\usepackage{amsmath}
\usepackage{amssymb}
\usepackage{color}
\makeatletter

\usepackage{psfrag}

\usepackage{babel}

\makeatother

\begin{document}
\newlength{\singfigwidth}
\setlength{\singfigwidth}{1.0\columnwidth}
\newlength{\doublefigwidthone}
\setlength{\doublefigwidthone}{6.5in}
\newlength{\doublefigwidth}
\setlength{\doublefigwidth}{6.5in}

\title{Minimal spanning trees at the percolation threshold: a numerical calculation}
\author{Sean M. Sweeney}
\author{A. Alan Middleton}
\affiliation{Department of Physics, Syracuse University, Syracuse, NY 13244, USA}

\begin{abstract}
	The fractal dimension of minimal spanning trees on percolation clusters
	is estimated for dimensions $d$ up to $d=5$.  
	A robust analysis technique is developed for correlated data, as seen
	in such trees.  This should be a robust method suitable for 
	analyzing a wide array of randomly generated fractal structures.
	The trees analyzed using these techniques are built using a
	combination of Prim's and Kruskal's algorithms for finding minimal spanning trees.
	This combination reduces memory usage and allows for simulation of larger systems
	than would otherwise be possible.
	The path length fractal dimension $d_{s}$ of MSTs on critical percolation clusters 
	is found to be compatible with the predictions of the perturbation expansion 
	developed by T.~S.~Jackson and N.~Read 
	[T.~S.~Jackson and N.~Read, Phys.\ Rev.\ E \textbf{81}, 021131 (2010)].
\end{abstract}
\maketitle

\section{Introduction to the problem}\label{Introduction}

	The statistical physics and dynamics of disordered physical
	systems naturally leads to the study of fractal geometric
	objects. Physical systems with quenched disorder, i.e., those with  
	fixed random heterogeneities, often have power law
	correlations at large scales or interfaces that are fractal or self-affine.
	These structures can result from some global optimization problem that
	has connections with graph theory.  
	For example, Dijkstra's shortest path algorithm \cite{Cormen-et-al, Dijkstra} can be used to 
	find the lowest energy path of a vortex line in a disordered superconductor \cite{Petaja-et-al, Huse-Henley},
	and these paths are self-affine.
	Excitations in a disordered XY model in two dimensions \cite{Zheng-et-al, Middleton}, 
	domain walls in spin glasses \cite{Bray-Moore}, and boundaries between drainage
	basins \cite{Herrmann-et-al} are examples of physical objects with
	fractal dimension that are found by global optimization.
	
	A structure with fractal scaling that arises in physical
	contexts is the minimal spanning tree (MST). One such context is 
	a highly disordered Ising spin glass. 
	Newman and Stein showed that for the strongly disordered limit, 
	the problem of finding a ground state can be directly mapped to finding an MST \cite{Newman-Stein}.
	This mapping can be used to investigate the multiplicity of ground 
	states in the thermodynamic limit.
	Minimal spanning trees have other applications such as in transportation
	networks connecting cities \cite{Prim}, telecommunications networks 
	connecting remote computer terminals \cite{Chou-Kershenbaum}, efficient
	circuit design \cite{Loberman-Weinberger}, taxonomic reconstruction of 
	evolutionary trees \cite{Edwards-Cavalli},
	and pattern recognition in image analysis \cite{Osteen-Lin}.

	A minimal spanning tree is a structure that connects a set of nodes with
	minimum total cost.  This structure is defined for a weighted graph $G=(V,E,w)$ where $V$ is a set of 
	vertices (or nodes), $E$ is a set of edges that connect vertices, and
	$w$ is a weight function, with each edge $e \in E$ having weight $w(e)$.
	A spanning tree is a loopless 
	connected set of edges that includes all the vertices in $V$.  The minimal 
	spanning tree is the spanning tree $T$ that minimizes the total weight
	\begin{equation}
		\label{eqn:MST-defintion}
		w(T) = \sum\limits_{e \in T} w(e) \ .
	\end{equation}
	This is a well-known problem in computer science and combinatorics.
	See Ref.~\cite{Algorithms} for a general overview of MSTs and MST-finding
	algorithms.

	A notable fact about the MST is that the minimal tree is determined only by
	the numerical ordering of the weights, i.e., it is otherwise independent of their value.
	So there is a large invariance or universality for these structures; their
	geometry is independent of the disorder distributions. As long 
	as the weights $w(e)$ are independently drawn from the same distribution 
	(independent identically distributed or i.i.d. weights), the edges
	can be sorted in order of increasing weight.  This ordering alone determines
	the tree.  Two physical problems with wholly different distributions of 
	quenched disorder have MSTs with equivalent statistics.  
	
	Recently, Jackson and Read carried out an analytical calculation for the
	fractal dimension $d_{s}$ of paths in MSTs \cite{Jackson-Read-I,
	Jackson-Read-II}. They developed a perturbation expansion
	for $d_{s}$ for MSTs on critical percolation clusters in $d$ dimensions, 
	obtaining the result
	\begin{equation}
		\label{eqn:perturbation-expansion}
		d_{s} = 2 - \frac{\epsilon}{7} + \mathcal{O}({\epsilon}^{2}) \ ,
	\end{equation}
	where $\epsilon = 6 - d$ and $d=6$ is the upper critical dimension
	\cite{Jackson-Read-II}.
	In general, disordered systems are difficult to analyze and rarely yield 
	quantitative analytic results.  This prediction therefore provides a strong motivation 
	for more precise computation of fractal dimensions in spanning trees, in 
	particular, dimensions of the trees on spanning percolation clusters.
	We note that it is unclear whether the fractal dimension of these trees is affected
	by being constructed on spanning percolation clusters as opposed to
	a whole lattice.
	The work presented in this paper seeks to numerically compute $d_s$ in
	dimensions $2 \le d \le 5$ for comparison with Eq.~(\ref{eqn:perturbation-expansion}).
	This calculation employs a combination of memory-saving techniques to simulate 
	large systems as well as careful data fitting procedures to obtain precise 
	estimates for $d_{s}$ in the limit of infinite system size.

	Our final calculations for $d_s$ yield values of 1.216(1) for $d=2$, 1.46(1) for $d=3$, 
	1.65(2) for $d=4$, and 1.86(4) for $d=5$ (refer to Table \ref{table:ds}).
	We develop and utilize a ${\chi}^2$ test that accounts for the scale-invariant correlations
	found in the data, allowing for an improved ${\chi}^2$ measure and robust estimates of the
	uncertainty in the effective exponent at scale $L$, $d_{s}(L)$.  We then use linear
	and nonlinear least squares fitting to extrapolate to the infinite system size
	limit.
	We find the numerical results to be of higher precision than previous calculations
	and compatible with Eq.~(\ref{eqn:perturbation-expansion}), though more conclusive confirmation
	requires improved numerical statistics and higher order analytic work.
	We emphasize that the analysis procedure used here is generalizable and could be
	useful for other work dealing with disordered systems.

\section{Model \& Algorithms}\label{Algorithms}

	To model MSTs on percolation clusters, we simulated hypercubic lattices 
	with $L$ vertices per side with periodic boundary conditions, giving 
	$L^{d}$ total vertices and $d{L}^{d}$ total edges. 
	Edges are independently given a weight randomly distributed on the interval $(0,1)$,
	where $w(e)$ is represented by a double precision number.  Very rarely two edges are assigned 
	the same weight.  In this case, a new random weight is generated for one of these
	edges until a weight is generated that does not match any previously assigned weight.
	An important distinction to note is that rather than seeking to find a tree that 
	spans all of the vertices in these hypercubic lattices, we seek to find 
	the MST on a percolation cluster that wraps around the periodic 
	lattice (in any of the $d$ directions) at the threshold of percolation. 
	Thus the MST-finding algorithms are stopped when the tree contains a subset of vertices 
	that wraps around the system instead of including all $L^{d}$ vertices of the graph.
	The final object of interest, the MST on a critical percolation cluster, 
	contains only a small subset of the set of vertices $V$ in the lattice.
	
	In order to avoid confusion, we first note a dual usage of the term spanning.
	For an MST, spanning means that the tree includes all vertices
	in the graph for which the MST is found.  In the context of percolation
	theory, the term refers to a cluster that is spanning or percolating around 
	the lattice.  We use spanning in both senses, with the correct sense implied by the context.
	
	One naive approach to finding the MST on a graph is to iterate
	through the list of all spanning trees and 
	select the tree with the lowest total weight.  This approach
	might work on a small finite graph, but
	the number of trees grows exponentially with $L^{d}$.
	In order to analyze properties of MSTs in the thermodynamic limit of infinite 
	system size, a more efficient MST-finding algorithm (both in terms of running 
	time and maximum memory requirements) is needed.

	As background for the algorithms used to find MSTs, it might be helpful 
	to briefly review the relationship between invasion percolation and Bernoulli percolation.
	In Bernoulli (bond) percolation \cite{Broadbent-Hammersley,Frisch-Hammersley}, 
	edges in a random graph have a probability $p$ to be 
	occupied, and a probability $1-p$ to be unoccupied.  
	After determining the occupation of each edge independently, one 
	inspects the graph to check for long-range connectivity in the form of a cluster of 
	connected vertices that percolates, i.e., spans the graph.  
	On a periodic hypercubic lattice, one definition of a percolating
	or spanning cluster is a cluster that wraps around the 
	lattice along one or more of the $d$ spatial dimension axes.  
	Such a cluster contains a loop that cannot be contracted to a point.  

	Examining larger and larger systems on a macroscopic scale, this percolation transition 
	becomes clear; below some critical percolation probability $p_{c}$ only small clusters 
	are seen, but at $p = p_{c}$ large clusters that span the system begin to emerge.  Aizenman
	refers to these large critical clusters as incipient spanning clusters (ISCs),
	a term which is closely related to, but distinct from, the incipient infinite cluster 
	(IIC) \cite{Aizenman-I, Aizenman-II}.
	In the exploration of Bernoulli percolation, this occupation 
	probability $p$ is finely tuned in order to observe this critical 
	transition and the clusters that are formed at criticality.

	An alternate approach to percolation is invasion percolation 
	\cite{Wilkinson-Willemsen}.  Invasion percolation is a procedure
	that greedily occupies edges of low weight $w$ and has a termination condition.
	Invasion percolation consists of a growing cluster $C$ that is a subset of $E$.
	This cluster has a set of adjacent edges $\partial C$.  The cluster $C$
	begins as a single occupied seed site.  Additional sites are ``invaded'' 
	by choosing from $\partial C$ the lowest weight edge, $e_{l}(\partial C)$
	and expanding the invaded region to include this edge, by adding $e_{l}(\partial C)$
	to $C$.  This invasion percolation process can be repeated until 
	long-range connectivity is observed (i.e., until the invaded region
	percolates across the system).

	There has been much discussion \cite{Goodman,Damron-Sapozhnikov-Vagvolgyi,Berg-et-al} 
	on how to relate the clusters of invasion
	percolation to those of ordinary Bernoulli percolation.
  There seems to be a strong reason to believe that in the limit 
	of infinite system size, infinite connected clusters 
	created using both of these percolation methods should obey 
	the same scaling relations, meaning that the fractal path length 
	dimension $d_{s}$ should be the same for both methods.  
	Invasion percolation is extremely useful because it allows
	for the simulation of critical percolation systems without
	having any knowledge of what the value of $p_{c}$ is for
	a particular system.  Thus invasion percolation is an 
	example of self-organized criticality 
	\cite{Goodman,Damron-Sapozhnikov-Vagvolgyi}.  

	In our particular case, we're interested in analyzing the
	MSTs that lie on an ISC.  In other words, we want to 
	find the MST for the ISC.  We'll refer to this final
	object as the MTISC (minimal tree for an incipient spanning
	cluster).  Here we use two algorithms, Kruskal's algorithm \cite{Kruskal} and Prim's 
	algorithm \cite{Prim}.  While Prim's algorithm is similar 
	to the invasion percolation process described above \cite{Barabasi}, 
	Kruskal's algorithm is related to Bernoulli 
	percolation due to the global nature of the algorithm.  However, 
	neither algorithm actually requires the choice of an occupation probability $p$ 
	as a parameter, and consequently both processes exhibit self-organized critical behavior,
	as the growth is stopped when a tree is found that wraps \cite{Machta-Chayes}.
	For a more precise definition of these algorithms, refer to Appendix A.
	We next present less formal descriptions of these algorithms.

\subsection{Kruskal's algorithm}

	Kruskal's algorithm is an MST-finding algorithm that considers all
	edges of a graph.  In Kruskal's algorithm,
	we grow a forest of many small trees, merging small groups of
	trees into larger trees and eventually identifying one
	of these large trees as the MTISC, terminating the algorithm.  	
	At the start of Kruskal's algorithm, each vertex is its own 
	isolated tree.  
	All of the edges that are not yet part of a tree are sorted,
	and the edge with minimal weight is selected, excluding edges
	that would form non-wrapping loops.  When an edge is selected,
	the trees containing each of its vertices are joined into a single tree.
	This process continues until a tree grows big enough that it wraps around 
	the periodic lattice.  At this time, this tree is identified as the Kruskal's MTISC, $T_{K}$.
	This process is illustrated in Fig.~\ref{fig:kruskals_growth}.

\begin{figure}
	\includegraphics[%
		width=0.5\singfigwidth, 
		keepaspectratio]{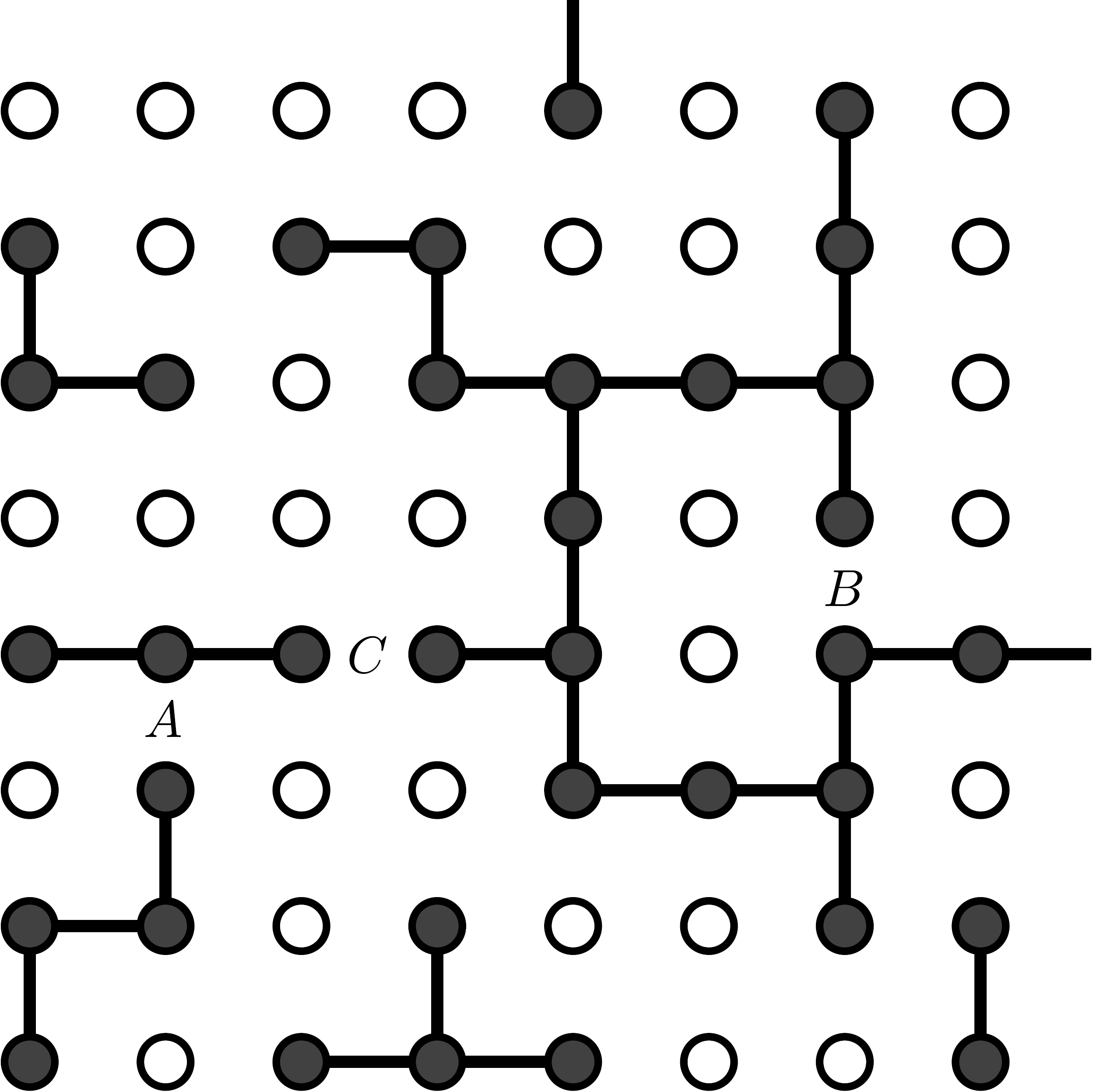} 
	\caption{\label{fig:kruskals_growth} A sample iteration of Kruskal's algorithm
		on an eight by eight periodic square lattice ($d=2$). At a given step, the
		state is a forest of trees, which includes isolated sites (open circles)
		and larger trees (connected solid circles).  During each step of the 
		algorithm, the edge with the lowest weight is selected from the remaining 
		unselected edges. For example, if edge $A$ is selected, the trees containing 
		either endpoint of edge $A$ are merged into a single tree.  If edge $B$ is 
		selected, its addition would form a non-wrapping loop (forbidden cycle), so edge $B$ is not 
		added to any tree and is removed from future consideration.  If edge $C$ is selected, 
		its addition would form a wrapping loop (allowed cycle), so the algorithm is terminated.
		The tree containing the endpoints of edge $C$ in then the Kruskal's MTISC, $T_{K}$.
	}
\end{figure}

	Kruskal's algorithm bears similarity to Bernoulli
	percolation in that it is a non-local algorithm that requires information
	about all edges of the graph. If we consider a graph with edges whose 
	weights are distributed evenly on $(0,1)$ and run Kruskal's algorithm 
	until we first examine an edge with weight $w > p_{c}$, the sites of
	$T_{K}$ are identical to those obtained from performing Bernoulli 
	percolation on the same graph with occupation probability $p_c$.

\subsection{Prim's algorithm}

	Prim's algorithm is equivalent to algorithms for loopless invasion 
	percolation \cite{Barabasi}.  At the start of the algorithm, the MST consists of a
	single seed site.  This tree grows outward from this vertex	through 
	the examination and conditional addition of adjacent edges, as
	edges adjacent to the growing tree are examined.  At any given iteration, 
	the minimal weight adjacent edge is selected 
	and incorporated into the tree, excluding edges that would lead to 
	non-wrapping loops.  The check for which edges form a loop is simple:  if
	both ends of the edge are in the current tree a loop would be formed.
	To check whether a loop is wrapping or non-wrapping, the algorithm assigns 
	to each vertex $a$ a displacement $\vec{r}_{a}$ from the origin.  
	This displacement is found for a newly added vertex $b$ by adding a vector 
	displacement $\vec{e}_{ab}$ for the edge to the vector displacement of the 
	end of the edge (site) that is already in the tree $\vec{r}_{a}$.  Thus for 
	a newly added site $b$ we have $\vec{r}_{b} = \vec{r}_{a} + \vec{e}_{ab}$.
	When a loop is encountered, this new site $b$ will already have a vector 
	displacement $\vec{r}_{b}^{\,\prime}$ defined, since this site is already in the tree.
	If the relative displacement $\left|\vec{r}_{b} - \vec{r}_{b}^{\,\prime}\right|$
	between these two labelings is greater than $L$, the loop formed is a wrapping loop.  
	When this wrapping edge is examined, the algorithm is terminated, and this 
	final tree is identified as the Prim's MTISC, $T_{P}$, as shown in Fig.~\ref{fig:prims_growth}.

\begin{figure}
	\includegraphics[%
		width=0.5\singfigwidth, 
		keepaspectratio]{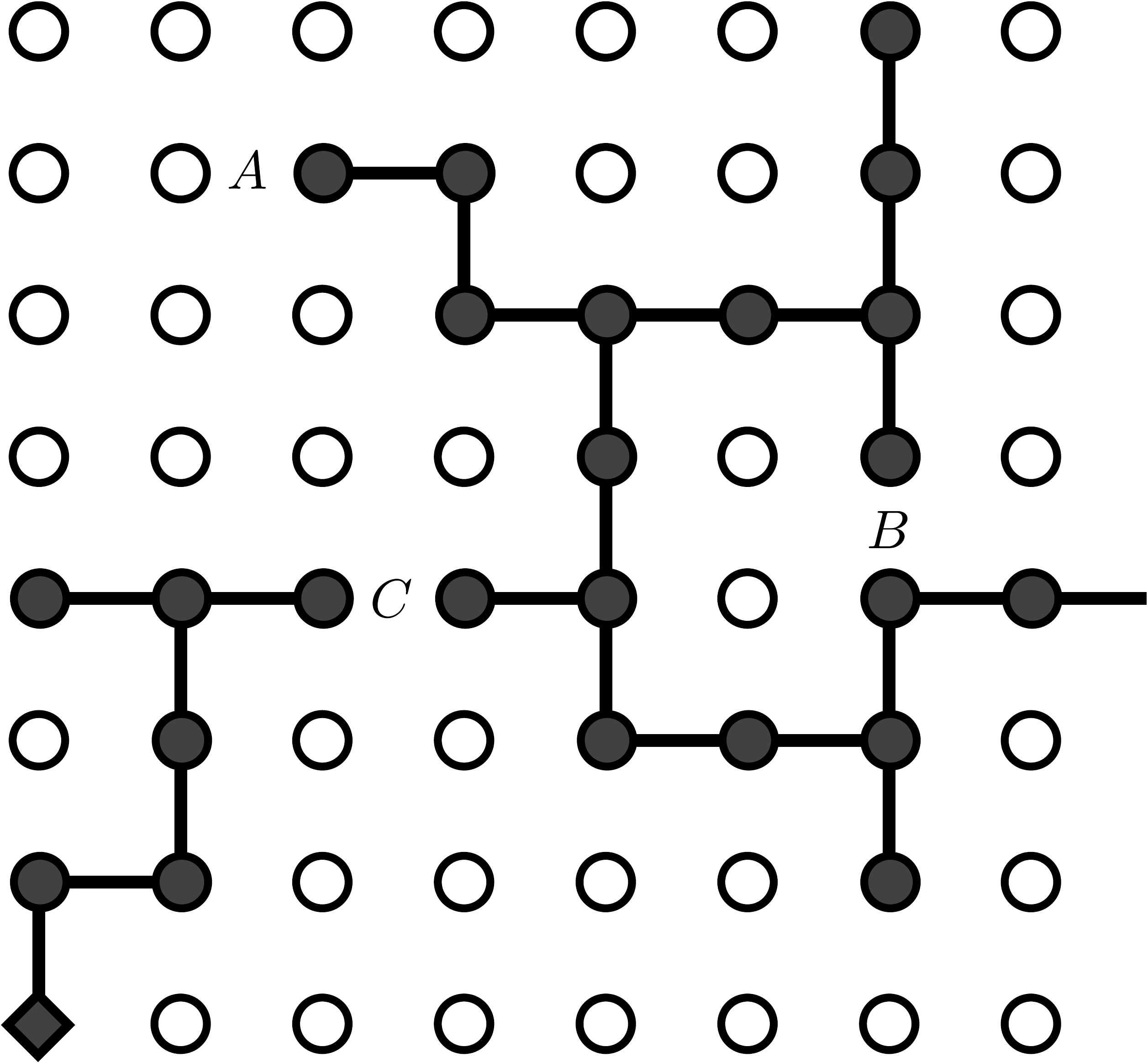}
	\caption{\label{fig:prims_growth} A sample iteration of Prim's algorithm
		on an eight by eight periodic square lattice ($d=2$). The tree is represented
		by connected solid circles and lines, while unoccupied sites are shown as open circles.  
		The diamond-shaped vertex at the bottom left of the lattice is the initial seed site $v_0$. 
		During each step of the algorithm, the edge adjacent to the current Prim's tree with 
		the lowest weight is selected.  If edge $A$ is selected, it is added 
		to the Prim's tree, along with the endpoint of $A$ that is not already 
		in the Prim's tree. If edge $B$ is selected, its addition to the Prim's 
		tree would form a non-wrapping loop (forbidden cycle), so edge $B$ is not added to the tree 
		and is removed from future consideration.  If edge $C$ has the lowest weight
		and is selected, its addition would cause the Prim's tree to wrap around 
		the periodic lattice (adding an allowed cycle), so the algorithm is terminated, 
		leaving the current Prim's tree as the Prim's MTISC, $T_{P}$.
	}
\end{figure}

	To increase efficiency when generating MTISCs, Kruskal's and Prim's algorithms can be 
	engineered to save memory by creating edges and weights only as needed during the execution 
	of the growth algorithm rather than at the start \cite{Paul-Ziff-Stanley, Grassberger}.  
	The memory saved in the Prim's method is much larger and leads to an overall more
	memory efficient method because Prim's algorithm only requires 
	the growing of a single tree rather than a large forest of trees.  
	Due to this decreased memory usage, it is possible to simulate larger
	systems with Prim's algorithm than with Kruskal's algorithm.

\subsection{Two-step method}

	We used a two-step method, combining Prim's and Kruskal's algorithms
	in order to take advantage of the increased efficiency afforded by Prim's algorithm
	when selecting a typical MTISC.
	We first generate a tree $T_{P}$ using Prim's algorithm and then apply Kruskal's
	algorithm to this tree, ensuring that the MTISC obtained has the same scaling
	as the MTISC that would be obtained by increasing $p$ large 
	enough to form a forest, one of whose trees is the final tree $T_{K}$.  	
	The data presented in this paper comes from analysis of MTISCs 
	generated by the two-step method, as well as comparisons with the intermediate
	data from $T_{P}$, before Kruskal's algorithm is applied to $T_{P}$.

	While Prim's algorithm in general takes less time and memory to find an MTISC than
	Kruskal's algorithm, the two algorithms do not find exactly the same MTISC \cite{Dobrin-Duxbury-II}.  
	As shown in detail in Appendix B, either $T_{K}$ is a subset of $T_{P}$, 
	or they do not intersect.  We are interested in constructing an MTISC that 
	is either identical to $T_{K}$ or scales the same as $T_{K}$.  The non-local 
	greedy edge selection of Kruskal's algorithm guarantees $T_{K}$ to have the 
	same sites as a Bernoulli percolation cluster.

	A technical detail to note about the two-step method is that in the final stage of 
	Prim's algorithm when a wrapping edge is selected, we add this edge to the final
	Prim's tree $T_{P}$.  This is done so that when Kruskal's algorithm is applied to
	$T_{P}$, there will be a wrapping edge to satisfy the termination condition 
	of Kruskal's algorithm.  While the inclusion of this edge temporarily destroys 
	the tree-like nature of $T_{P}$, its inclusion is essential.  

	Because Prim's growth begins at a random seed site that does not necessarily
	belong to $T_{K}$, often the $T_{K}$ is a subset
	of $T_{P}$. This of course depends which vertex $v_0$ is used as the seed site for
	Prim's growth, as the Prim's MTISC is a function of its seed site, $T_{P}(v_0)$.  
	Executing Kruskal's algorithm in the second step of this method 
	serves to ``trim off'' the part of $T_{P}(v_0)$ that includes the seed 
	site $v_0$ and that does not belong to $T_{K}$. The tree 
	derived from this procedure is termed the two-step MTISC, $T_{2}$.
	This method is illustrated in Fig.~\ref{fig:twostep_growth}.

	We show in Appendix B that this two-step procedure yields the Kruskal's
	MTISC, i.e., $T_{2}=T_{K}$, in all cases except when 
	there naturally arise multiple disjoint ISCs on the lattice at 
	criticality.  By direct simulation of these systems, we observed 
	that $T_{2} \ne T_{K}$ about 0.2\% of the time in 
	two dimensions, 1\% in three dimensions, and 5\% in four dimensions. 
	Five dimensional samples were not compared due to the large memory 
	demands of simulating minimal spanning forests in five dimensions.
	
	We suppose by standard scaling that in the case where $T_{2} \ne T_{K}$, 
	$T_{2}$ and $T_{K}$ should have the same scaling properties.  We have 
	simulated some samples in this case for systems up to size $512^2$, 
	$64^3$, and $32^4$ and observed similar scaling properties between 
	$T_{2}$ and $T_{K}$.  The memory requirements for the two-step method 
	allowed us to simulate systems of size $2048^2$ within roughly 2GB of 
	memory, $256^3$ within 1.5GB, $64^4$ within 1GB, and $48^5$ within 2GB.

\begin{figure}
	\begin{center}
		\subfigure[]{%
			\label{fig:twostep_start}
			\includegraphics[%
				width=0.5\singfigwidth, 
				keepaspectratio]{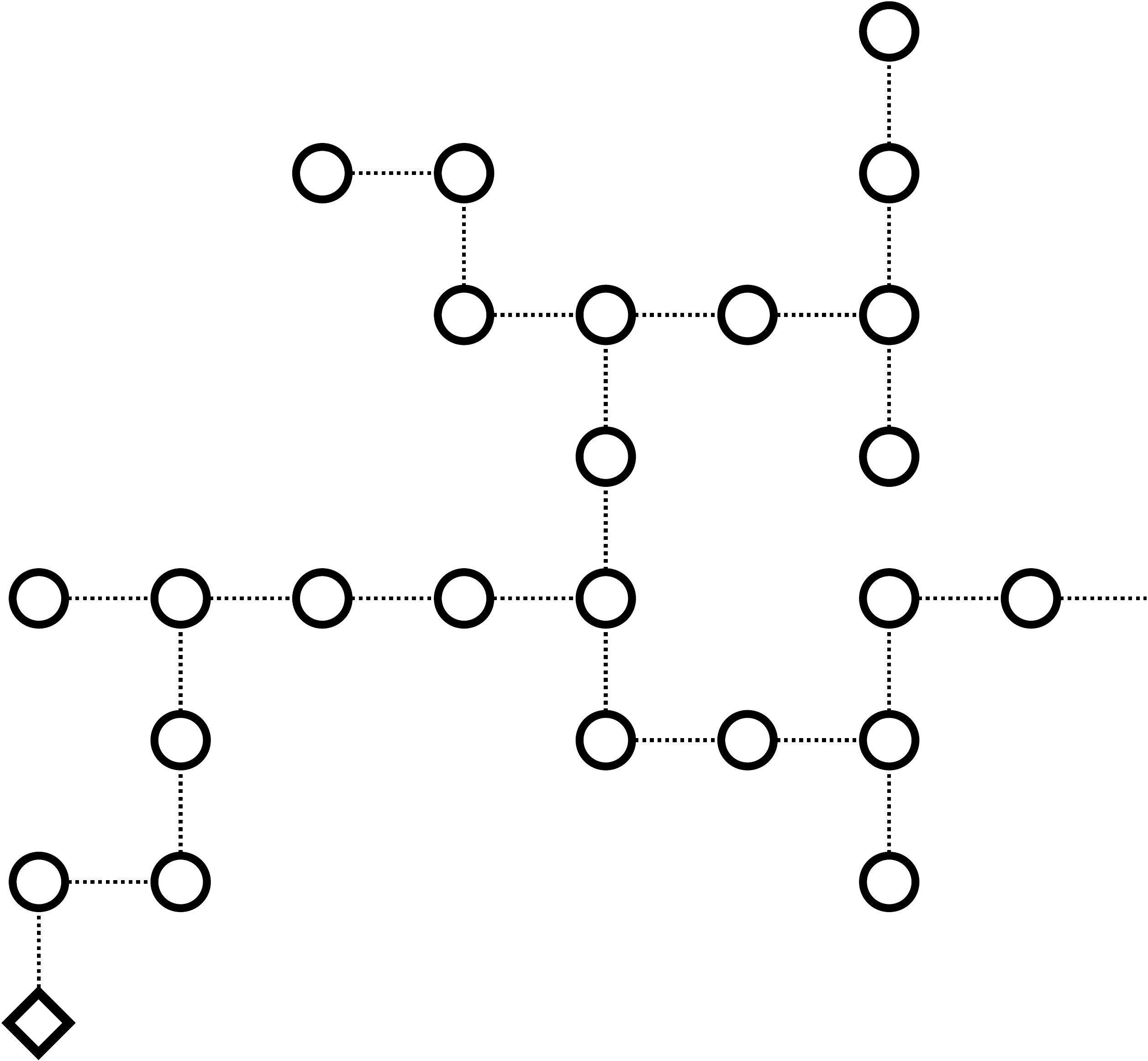}
		}%
		\subfigure[]{%
			\label{fig:twostep_intermediate}
			\includegraphics[%
				width=0.5\singfigwidth, 
				keepaspectratio]{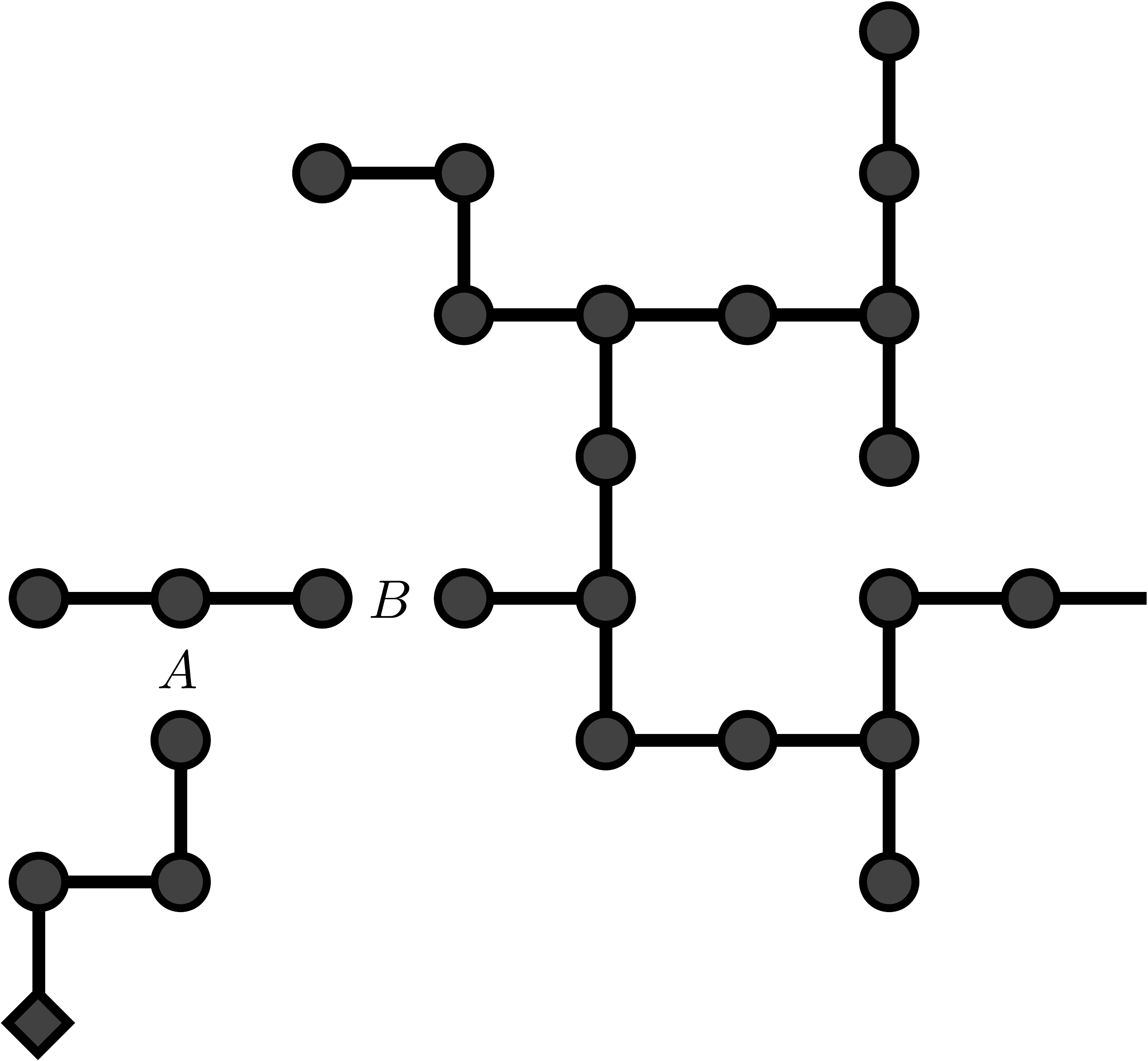}
		}%
	\end{center}
	\caption{Illustration of the two-step method.  This method grows a 
		candidate MTISC using the Prim's method and then trims the tree using
		Kruskal's algorithm.  (a) depicts $T_{P}$ 
		(plus the final wrapping edge) which is given as
		input to Kruskal's algorithm. During construction of this tree,
		only edges adjacent to the tree are tested.  The potential connectivity 
		of the unoccupied sites is explicitly shown using open circles and thin 
		lines, and for reference the Prim's seed site $v_0$ at the bottom left of the 
		lattice is represented by a diamond.  (b) shows the state of the graph after 
		running 23 steps of Kruskal's algorithm on $T_{P}$, with trees being represented by
		connected solid circles. If edge $B$ is selected before edge $A$, i.e., 
		$w(B)<w(A)$, the Prim's seed site $v_0$ will not be part of $T_{K}$, and the 
		disconnected portion of the graph containing the Prim's seed site $v_0$ will be 
		trimmed off.
	}
	\label{fig:twostep_growth}
\end{figure}

\section{Analysis/Methods}

\subsection{Methods for scaling analysis}

	To find the fractal dimension of the minimal trees on spanning percolation clusters,
	we compute the Euclidean distance $r$ and path length $s$ between some origin on the MTISC
	and other points on the MTISC.  Accurately determining the scaling in the limit of large
	clusters requires taking into account lattice effects, finite size effects, and
	statistical uncertainties.
	Given any two points on a tree, there is a
	unique path between the two points, so that the path length is easily defined.
	The choice of endpoints for the paths is chosen in a natural fashion for each
	tree.  
	
	For trees constructed using Prim's algorithm, the origin is taken to be 
	the seed site where cluster growth begins.  The set of paths connecting the origin
	to all other points in the tree is used in the statistics.  For each tree $T_{P}$, we find
	$n_{P}(s)$, the number of paths of length $s$ that start at the origin.  We use 
	$r_{P}(s)$ to indicate the average over all paths of length $s$ of the Euclidean distance
	$\left|\vec{r}\right|$. 
	
	After then running Kruskal's
	algorithm on the Prim's tree to find the trimmed MTISC $T_{2}$, a random
	origin is chosen on the trimmed tree.  All paths from this new origin are used to find
	both $r(s)$ and $n(s)$, the averaged Euclidean distance and number of paths.  
	In order to reduce the amount of data stored,
	the full set of $N$ samples is grouped into sets of $N_b$ batches of uniform size $N/N_b$. 
	The	batch-averaged quantities of $r(s)$ and $n(s)$ for these trimmed trees are calculated and stored.  
	Here we use $r(s)$ and $n(s)$ to refer to data generated by the two-step algorithm.  
	As a comparison, we also looked at the averages for Prim's trees, before trimming.

	We assume that the paths on the tree are well described as fractal. The 
	scaling of the sample-averaged $r(s)$ will then follow the relation
	\begin{equation}
		r(s) \sim s^{1/{d_{s}}} \ .
		\label{eqn:scaling-rs}
	\end{equation}
	If we write $L$ as the
	length of one side of a hypercubic system, then $s/L^{d_{s}}$ is a natural scaling
	parameter, and we expect to see finite size effects near $s/L^{d_{s}} = 1$, 
	as paths of this length approach the size of the system.  
	The standard one parameter finite size scaling assumption is then
	that $r(s)$ will scale like $s^{{1}/{d_{s}}}$ multiplied by
	some unknown function of the argument ${s^{{1}/{d_{s}}}}{L}^{-1}$.  
	The form of the scaling hypothesis is that
	\begin{equation}
		\begin{split}
		r(s) &\approx s^{1/{d_{s}}} f\left(\frac{s^{1/{d_{s}}}}{L}\right) \\
					 &\approx s^{1/{d_{s}}} g\left(\frac{s}{L^{d_{s}}}\right) \ ,
		\end{split}
	\end{equation}
	where the scaling functions $f(\omega)$ and $g(\omega)$ behave as $f(\omega) \approx c_{1}$ 
	for some constant $c_{1}$ for $\omega \ll 1$, $f(\omega) \approx 0$ as $\omega \rightarrow \infty$, 
	$g(\omega) \approx c_{2}$ for some constant $c_{2}$ for $\omega \ll 1$, 
	$g(\omega) \approx 0$ as $\omega \rightarrow \infty$.
	For more compact formulas, we define the scaled dimensionless variables
	\begin{equation}
		\rho = \frac{r(s)}{s^{1/{d_{s}}}} \ , \label{eqn:rho-definition}
	\end{equation}
	\begin{equation}
		\omega = \frac{s}{L^{d_{s}}} \ . \label{eqn:omega-definition}
	\end{equation}
\begin{figure}
	\includegraphics[%
		width=\singfigwidth, 
		keepaspectratio]{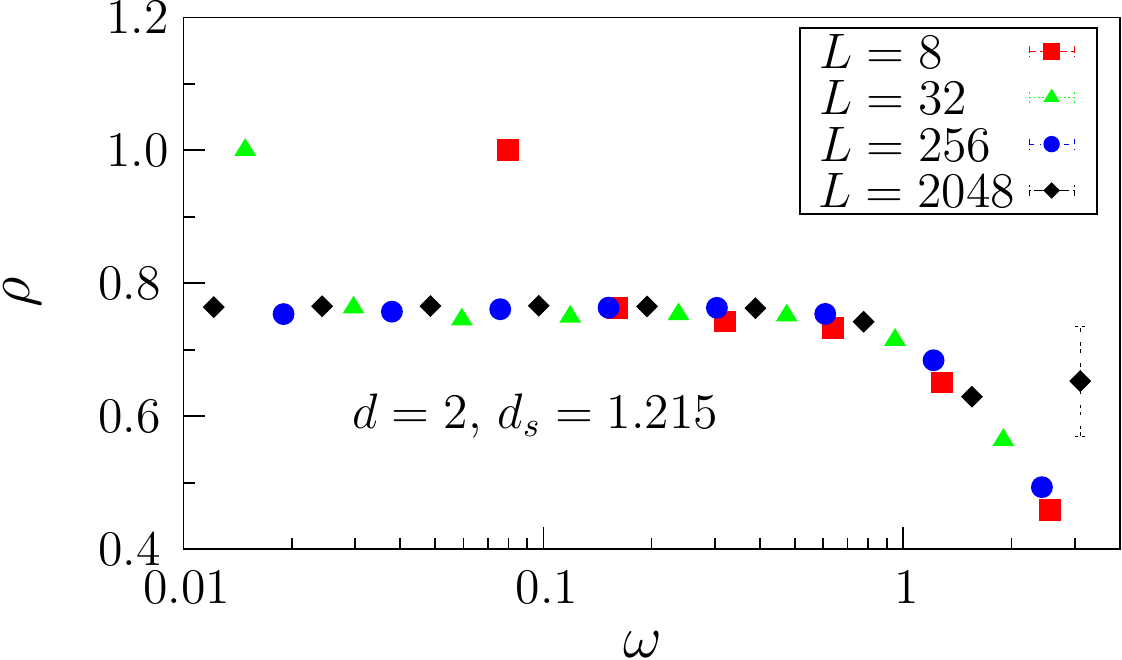}
	\caption{\label{fig:scaling-collapse} Scaling collapse for $d=2$
	with $d_{s}=1.215$, which appears to be moderately close to the 
	true value of $d_{s}$ judging by eye from the goodness 
	of the collapse.  Note that error bars are smaller than the symbol
	size for all points except the right-most point.
	}
\end{figure}
	
	Using this scaling hypothesis, we can make estimates for 
	$d_{s}$ by plotting data for multiple system sizes on the scaled
	axes of $\rho$ vs.~$\omega$.  If we tune the parameter $d_{s}$,
	we see that the curves for various sizes $L_{i}$ can be made to collapse well 
	near an estimated best value of $d_{s}$ (see Fig.~\ref{fig:scaling-collapse}).
	While this method allows us to get a fair idea of this exponent 
	$d_{s}$, it relies on subjective estimations of curve collapse
	and for that reason is not ideal when attempting precise
	estimates.
\begin{figure*}
	\includegraphics[%
		width=\doublefigwidth, 
		keepaspectratio]{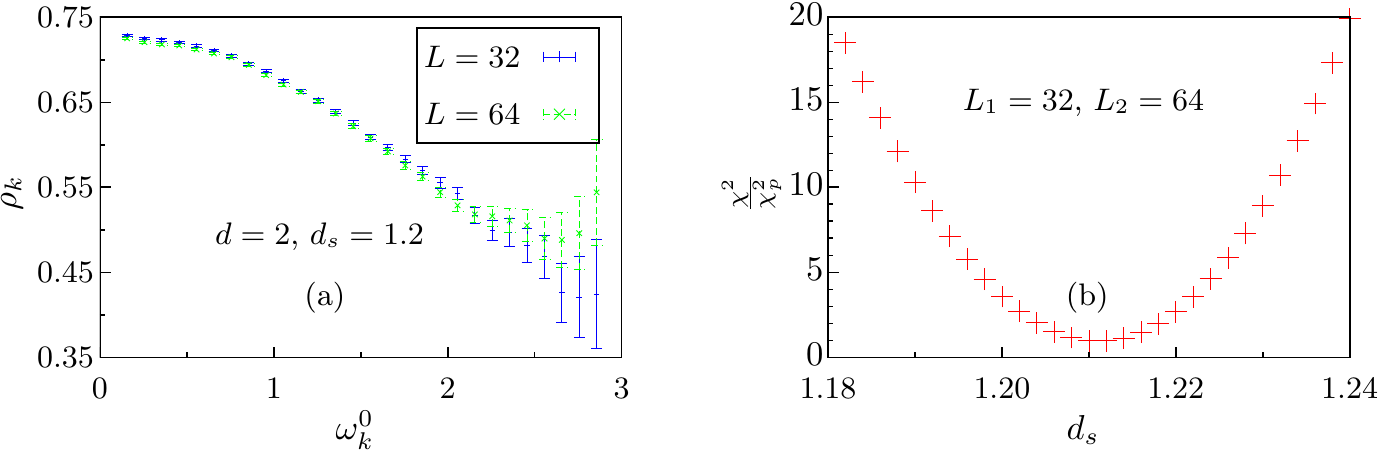}
	\caption{\label{fig:single-pair} A sample collapse for two systems of size $L=32$ and
		$L=64$ in dimension $d=2$. (a)
		Shows the comparison of the two systems at interpolated points 
		$\omega_{k}^0$ for a value of $d_{s}=1.2$. Comparing the value of 
		$\rho_{k}(L_{1}=32)$ and $\rho_{k}(L_{2}=64)$ at each of these points
		allows for the calculation of ${\chi}^{2}(d_{s};L_{1},L_{2})$.
		(b) Shows ${\chi}^{2}$ compared with an expected estimate ${\chi}_{p}^{2}$ 
		as a function of the fitting parameter $d_{s}$ for the same pair of systems.
		Though we determine final error bars by resampling, the ${\chi}^{2}$ model is 
		shown for comparison.
	}
\end{figure*}

	To have a better estimate for $d_{s}$ and to estimate the uncertainty in 
	this estimate, we have implemented an automated fitting procedure.
	This procedure determines an effective exponent $d_{s}(L)$ derived from data 
	for samples of size $L$ and of larger size $2L$ by minimizing a ``goodness of fit'' 
	parameter ${\chi}^{2}$ at each scale $L$.  The only input to this procedure is an
	estimate of $s_{l}$, the small path data cutoff, which is discussed in Appendix
	C.  We then extrapolate $d_{s}(L)$ for $L \rightarrow \infty$ to get our best estimate 
	for $d_{s}$.

	The key part of this calculation is to choose a robust and reliable
	measure for ${\chi}^{2}$.  This allows us to quantitatively 
	measure how well the data for a given pair of system sizes collapses as a function 
	of the parameter $d_{s}$. We seek a value of $d_{s}$ for which 
	this ${\chi}^{2}$ is minimized (Fig.~\ref{fig:single-pair}), 
	though the magnitude of our final error bars are determined by resampling.
	Our definition of ${\chi}^{2}$ must allow for 
	non-uniform correlations in fluctuations of $r(s)$ between 
	different values of $s$, discrete lattice effects 
	(small $s$ lower data cutoff), and statistical uncertainties 
	(large $s$ upper data cutoff).

\subsection{Correlations in $r(s)$}

	In order to define a useful ${\chi}^{2}$ statistic, we first focus on
	the correlations we observed in the $r(s)$ data
	for the spanning trees.  In summary, we find that these correlations
	have a range in $s$ that grows linearly with increasing path length $s$, 
	in dimensions $d=2,3,4,5$.  We then modify the standard ${\chi}^{2}$ test 
	for uncorrelated data to account for these correlations.

	To describe correlations in the averaged $r(s)$ data, it will be helpful to
	first define how our data is averaged over samples, since we use multiple 
	groupings of data for calculating averages.  Let $r_{i}(s)$ 
	denote an average of Euclidean distance $r=\left|\vec{r}\right|$ over all points on tree $i$ that 
	are at a chemical (path length) distance $s$ from the origin for each tree $i=1,\dots,N$ in the 
	$N$ samples.  For faster analysis, the $N$ samples are organized into $N_b$ batches.
	The batch index $\alpha$ ranges over $1\leq \alpha \leq {N_b}$.
	We will use $r_{\alpha}(s)$ to denote an average over the $m = N/{N_b}$ samples in
	batch $\alpha$:
	\begin{equation}
		\label{eqn:batch-avg-definition}
		r_{\alpha}(s) = \frac{1}{m}\sum\limits_{i \in \alpha} r_{i}(s) \ .
	\end{equation}
	The global average over all $N$ samples will be represented by 
	\begin{equation}
		\label{eqn:global-avg-definition}
		\overline{r(s)} = \frac{1}{N_{b}}\sum\limits_{\alpha = 1}^{N_b} r_{\alpha}(s) = \frac{1}{N}\sum\limits_{i=1}^{N} r_{i}(s) \ .
	\end{equation}
	
	To initially examine correlations over $s$ of $r_{i}(s)$ within a single tree, 
	we plot the fluctuations of the batch averages $r_{\alpha}(s)$ about the global average 
	$\overline{r(s)}$.  That is, we plot the variations of the average  ${\delta r}_{\alpha}(s)$, where
	\begin{equation}
		{\delta r}_{\alpha}(s) = r_{\alpha}(s) - \overline{r(s)} \ ,
	\end{equation} 
	vs.~path length $s$.  Fig.~\ref{fig:typical-batch-correlations} 
	displays these correlations for a typical batch of data in a system of size $64^2$.

	Note that the form of the correlations should be independent of batch size, 
	up to a multiplicative scaling factor.  We can see this explicitly by examining 
	$\overline{ {\delta r}_{\alpha}(s) {\delta r}_{\alpha}(t) }$ for path length values
	$s$ and $t$.	
	For $i \neq j$, $i$ and $j$ are independent samples, so we can use the relation
	\begin{equation}
		\overline{{\delta r}_{i}(s){\delta r}_{j}(t)} = \overline{{\delta r}_{i}(s)}\hspace{5 pt}\overline{{\delta r}_{j}(t)}
	\end{equation}
	and of course 
	\begin{equation}
		\overline{{\delta r}_{i}(s)} = \overline{r_{i}(s) - \overline{r(s)}} = 0 \ .
	\end{equation}
	By standard computation all of the $i \neq j$ cross terms are zero, and we are left with 
	\begin{equation}
		\label{eqn:batch-independent-correlations}
		\overline{ {\delta r}_{\alpha}(s) {\delta r}_{\alpha}(t) } = \frac{1}{m}\overline{ {\delta r}_{i}(s) {\delta r}_{i}(t) } \ ,
	\end{equation}
	showing that the choice of grouping the data into batches should not affect the form of the correlations
	in $r(s)$.

\begin{figure}
	\includegraphics[%
		width=\singfigwidth, 
		keepaspectratio]{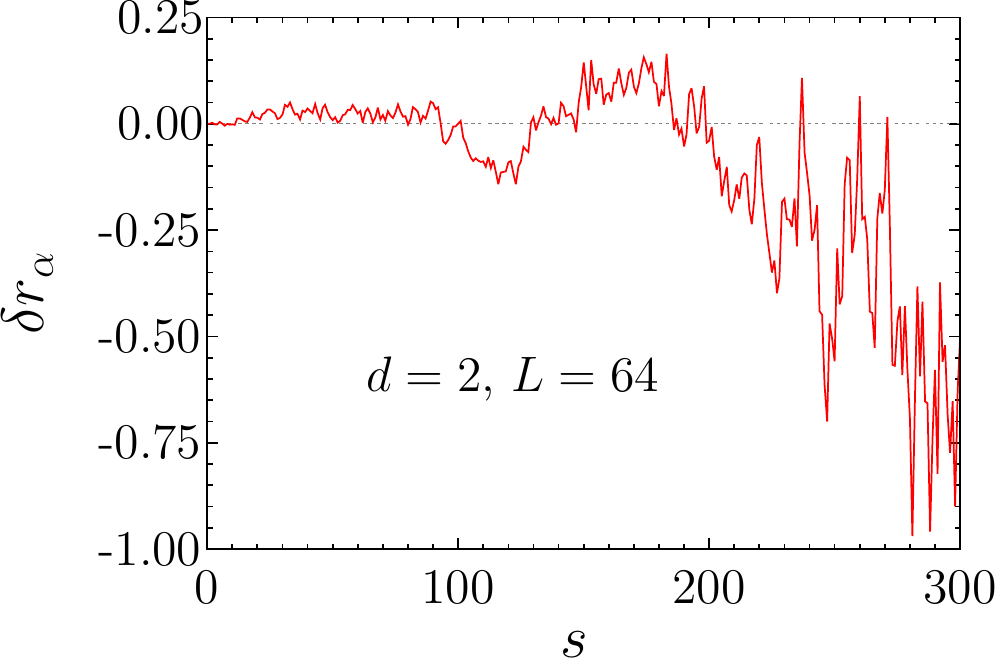} 
	\caption{\label{fig:typical-batch-correlations} Variations from the mean
		for a typical batch $\alpha$ of data for a system of size $L=64$
		in dimension $d=2$.  The difference between the mean Euclidean distances
		${\delta r}_{\alpha}(s) = r_{\alpha}(s) - \overline{r(s)}$ is plotted
		vs.~path length $s$.
	}
\end{figure}

	To quantitatively examine these correlations we compute the 
	correlation matrix $c_{s,t}$ defined as
	\begin{equation}
		\label{eqn:cij-definition}
		c_{s,t} = \frac{1}{N_b} \sum\limits_{\alpha = 1}^{N_b} { \frac{ {\delta r}_{\alpha}(s){\delta r}_{\alpha}(t) }{ {\gamma}(s) {\gamma}(t) } } \ .
	\end{equation}
	Here ${\gamma}(s)$ is the root mean square fluctuation 
	in ${\delta r}_{\alpha}(s)$ computed over the $N_b$ batches of data and is 
	used to normalize the entries $c_{s,t}$ of the correlation matrix:
	\begin{equation}
		\label{eqn:gamma-definition}
		{\gamma}^{2}(s) = \frac{1}{N_{b}} \sum\limits_{\alpha = 1}^{N_b} {\delta r}^{2}_{\alpha}(s) \ .
	\end{equation}

	A sample correlation matrix for $L=64$ is plotted in Fig.~\ref{fig:correlation-matrix}.
	The data suggests that the correlation length 
	increases with increasing $s$.  
	To construct a model for the scaled correlation length, we measured
	the width along the diagonal of the peak in the correlation matrix for 
	various values of $s$.  We used different measures 
	for measuring this width, including a measure of the 
	full width at half maximum (how many ``steps'' away from the diagonal
	before $c_{s,t}$ falls to a value of ${1}/{2}$), a measure of
	one decay length (how many steps from the diagonal until $c_{s,t}$
	drops to a value of ${1}/{e}$), and a measure using an exponential 
	decay model $c_{s,t} = \exp\left({-{\left|s-t\right|}/{\ell^\prime}}\right)$ 
	assuming an unscaled correlation length $\ell^\prime$.  
	We calculate $\ell^\prime$ using the zeroth moment of $c_{s,t}$, allowing
	us to then obtain the scaled correlation length
	$\ell = {\ell^\prime}/{L^{d_{s}}}$, given a rough estimate for $d_{s}$.
\begin{figure}
	\includegraphics[%
		width=\singfigwidth, 
		keepaspectratio]{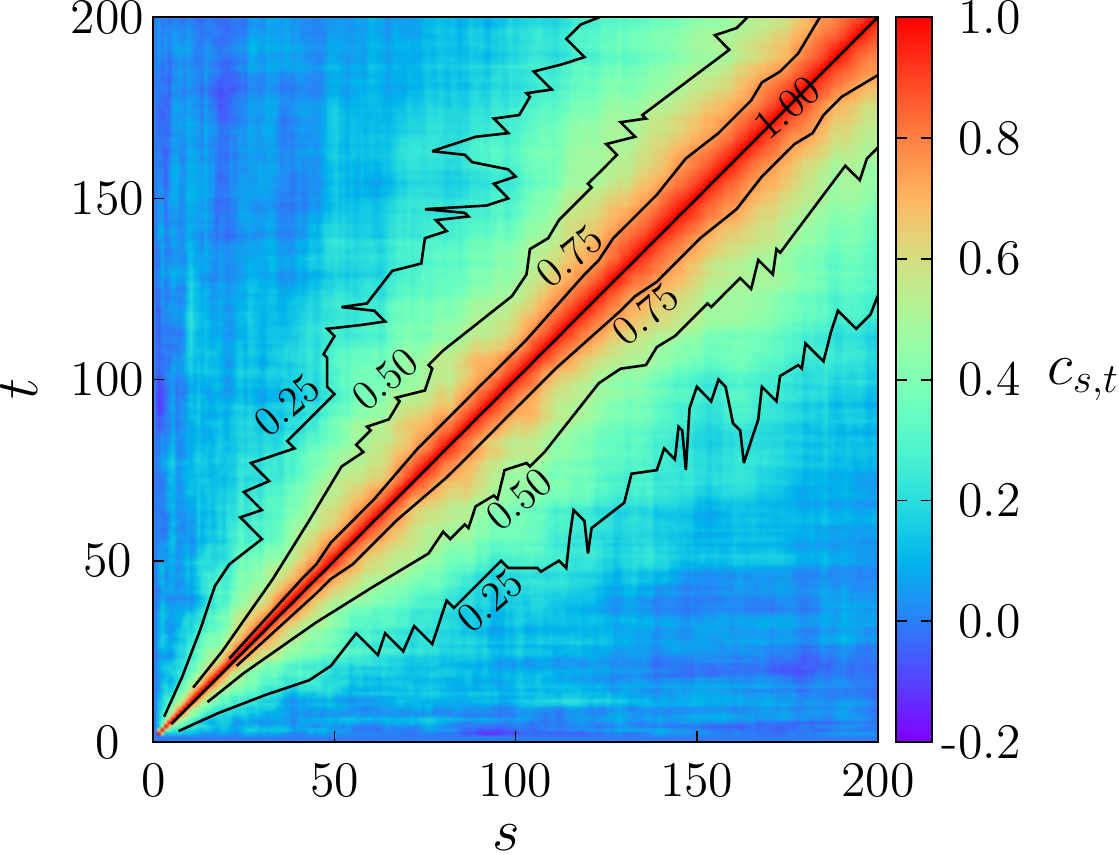}
	\caption{\label{fig:correlation-matrix} Correlation matrix
		$c_{s,t}$ averaged over all data for systems of size 
		$L=64$ in dimension $d=2$.  Selected contours are shown.
	}
\end{figure}

	Fig.~\ref{fig:correlation-length-collapse} shows the scaled correlation 
	length $\ell$ vs.~$\omega$ for multiple system sizes for dimension $d=2$,
	using the exponential decay model $c_{s,t} = \exp\left({-{\left|s-t\right|}/{\ell^\prime}}\right)$ 
	for the unscaled correlation length $\ell^\prime$.
	We see that $\ell$ increases linearly with $\omega$ 
	up until roughly $\omega = 1$, at which point the scaled correlation length 
	levels off to a constant value.  Although the proportionality
	constant differs depending on which method is used to measure the correlation
	length, we observed a linear relationship between $\ell$ and $\omega$ in all cases.
\begin{figure}
	\includegraphics[%
		width=\singfigwidth, 
		keepaspectratio]{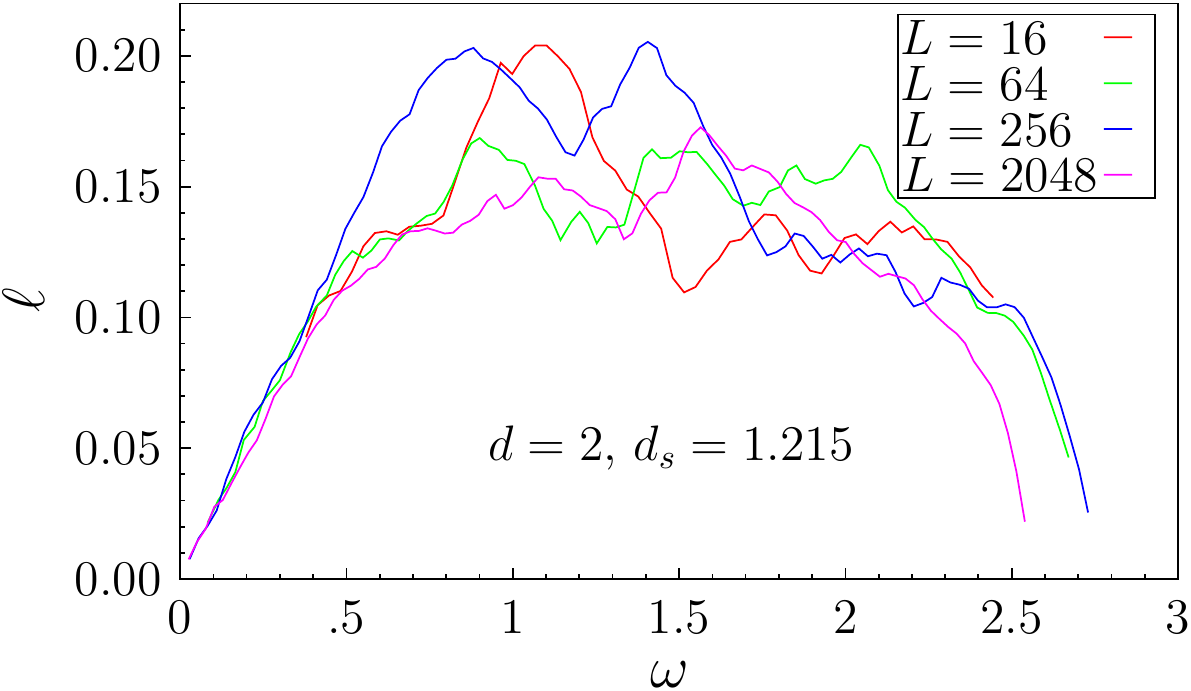} 
	\caption{\label{fig:correlation-length-collapse} Scaled correlation length
		$\ell$ vs.~$\omega$ for various systems in dimension $d=2$ with
		$d_{s}=1.215$.  In this case $\ell$ was measured using the exponential
		decay model $c_{s,t} = \exp\left({-{\left|s-t\right|}/{\ell^\prime}}\right)$ for the unscaled
		correlation length $\ell^\prime$, with the relation to the scaled correlation
		length $\ell$ being $\ell = {\ell^\prime}/{L^{d_{s}}}$.
	}
\end{figure}
	We also find this linearity in dimensions $d=2,3,4,5$.
	Based on this empirical observation, we choose an ansatz for the scaled correlation length.
	\begin{equation}
		\label{eqn:correlation-length-model}
		\ell(\omega) \propto \left\{
			\begin{array}{lr}
				\omega & : \omega \leq 1 \ \ \\
				1 & : \omega > 1 \ .
			\end{array}
		\right.
	\end{equation}
	For the exponential decay model used in Fig.~\ref{fig:correlation-length-collapse},
	the proportionality is ${\ell} \approx 0.24 \omega$.  Physically, this means that on average, for any given growing path 
	in a spanning tree, the path must grow about $24\%$ longer than its current length before
	the correlations in the $r(s)$ data for this path decay, using the exponential decay model
	$c_{s,t} = \exp\left({-{\left|s-t\right|}/{\ell^\prime}}\right)$ for the unscaled
	correlation length $\ell^\prime$.

	Now that we have a consistent model for the correlation length, we incorporate this
	model into a ${\chi}^{2}$ goodness of fit measure.  
	We use subscripts $1$ and $2$ to indicate data sets for systems of size $L_1$ and $L_2$ being
	compared.  We use $L_2 = 2 L_1$.
	A general	goodness of fit measure for quantifying how well two curves collapse 
	in variables $\rho$ vs.~$\omega$ starts with choosing a set of independent points 
	$\omega_{k}$ and at each point calculating the difference between $\rho_{1}(\omega_{k})$ 
	and $\rho_{2}(\omega_{k})$, $ \Delta_{12}(\omega_{k}) = \rho_{1}(\omega_{k}) - \rho_{2}(\omega_{k}) $.
	
	For a ${\chi}^{2}$ test with uncorrelated data this difference $\Delta_{12}(\omega_{k}^0)$ at the point ${\omega}_{k}^0$
	would then be squared and normalized by the sum of the variances, ${\sigma}_{1}^{2}(\omega_{k}^0)$ 
	for the system of size $L_{1}$ and ${\sigma}_{2}^{2}(\omega_{k}^0)$ for the system of size $L_{2}$.
	This gives 
	\begin{equation}
		\label{eqn:standard-chisquared-definition}
		{\chi}^{2}_{0}(d_{s};L_{1},L_{2}) = \sum\limits_{k} \frac{ {\Delta}_{12}^{2}(\omega_{k}^0) }{ {\sigma}_{1}^{2}(\omega_{k}^0) + {\sigma}_{2}^{2}(\omega_{k}^0) } \ ,
	\end{equation} 
	where we would sum over discrete points $\omega_{k}^0$ chosen with uniform spacing~\cite{Numerical_Recipes}.

	To incorporate the correlations we've observed
	into our definition of ${\chi}^{2}$, we compare the spacing
	of points chosen for the sum with the correlation length.  Specifically, we assume
	that if we choose points $\omega_{k}$ logarithmically spaced, i.e., $\omega_{k+1} = q \omega_{k}$,
	we will effectively have a constant scaled correlation length, for $\omega_{k} \leq 1$.  And for some
	choice of $q$ this scaled correlation length will be equal to $1$, allowing us
	to use the form of Eq.~(\ref{eqn:standard-chisquared-definition}):
	\begin{equation}
		{\chi}^{2} = \sum\limits_{k} \frac{ {\Delta}_{12}^{2}(\omega_{k}) }{ {\sigma}_{1}^{2}(\omega_{k}) + {\sigma}_{2}^{2}(\omega_{k}) } \ .
	\end{equation} 
	If we consider taking the continuum limit of this equation, we see that
	\begin{equation}
		\begin{split}
			{\chi}^{2} &= \int \frac{ {\Delta}_{12}^{2} }{ {\sigma}_{1}^{2} + {\sigma}_{2}^{2} }\,d(\log_q \omega) \\
								 &= \int \frac{ {\Delta}_{12}^{2} }{ {\sigma}_{1}^{2} + {\sigma}_{2}^{2} }\left(\frac{1}{\ln(q)\omega}\right)\,d{\omega} \\
								 &\propto \int \frac{ {\Delta}_{12}^{2} }{ {\ell}_{q}({\sigma}_{1}^{2} + {\sigma}_{2}^{2}) }\,d{\omega} \ ,
		\end{split}
	\end{equation}
	where we have defined ${\ell}_{q} = \ln(q)\omega$.  If we revert back to the discrete
	form, we see that our generalized goodness of fit measure becomes
	\begin{equation}
		{\chi}^{2} = \sum\limits_{k} \frac{ {\Delta}_{12}^{2}(\omega_{k}^0) }{ \ell(\omega_{k}^0) [{\sigma}_{1}^{2}(\omega_{k}^0) + {\sigma}_{2}^{2}(\omega_{k}^0)] } \ ,
		\label{eqn:custom-chisquared-definition}
	\end{equation}
	where $\ell(\omega_{k}^0) \propto {\ell}_{q}$ is the scaled correlation length at point
	$\omega_{k}^0$ \cite{Box-Pierce}. 
	Dividing by the scaled correlation length will effectively weight each scaled correlation 
	length sized ``box'' of data points in the ${\chi}^{2}$ sum as one independent 
	data point.  
	
	For our analysis, all runs ($N=4 \times 10^6$ samples per system size)
	were split into $N_b=400$ uniform batches (histograms) of $N/N_b = 10^4$ samples apiece.
	This was done in part because storing data for all $4 \times 10^6$ samples was
	impractical and running the bootstrap analysis over samples
	individually would be too time consuming to be feasible.
	To determine a useful batch size, we ran some preliminary 
	tests by varying the batch size and plotting scaling collapses like 
	Fig.~\ref{fig:scaling-collapse}.  We saw similar estimated values of
	$d_{s}$ based on these collapses, independent of batch size.  A batch
	size of 400 seemed balanced because it allowed for error bars of order 
	${1}/{\sqrt{N_b}} \approx 5\%$.  As we have shown in Eq.~(\ref{eqn:batch-independent-correlations}), 
	the form of correlations should be independent of batch size chosen.

\subsection{Extrapolation to $L=\infty$}

	Next we must consider exactly what region in $s$ of the data we want to 
	fit.  We address lattice effects by examining a lower (small $s$) data cutoff.
	The upper (large $s$) cutoff is also considered due to low statistics (large uncertainties)
	for $s \gg L^{d_{s}}$, though in the end we find that no upper cutoff is necessary.
	To determine reasonable cutoffs, we examine how the measured ${\chi}^{2}$ 
	and $d_{s}$ respond to changes in these cutoffs.  A more detailed discussion 
	is included in Appendix C.  Once we've decided upon fair cutoffs, 
	this collapse procedure is run $N_b$ times (once for each batch of data for each pair of
	sizes $L_1$, $L_2$ with $L_2 = 2 L_1$). Each time, 
	the value of $d_{s}$ for which ${\chi}^2$ is minimized is found, giving $N_b$ 
	independent estimates ${d_{s}}(\alpha,L_1,L_2)$ for a given pair of systems
	$L_1$ and $L_2$.  The final estimate of $d_{s}$ for this pair of systems 
	is the average of the $N_b$ independent estimates:
	\begin{equation}
		\overline{d_{s}(L_1,L_2)} = \frac{1}{N_b} \sum\limits_{\alpha = 1}^{N_b} {d_{s}}(\alpha,L_1,L_2) \ .
	\end{equation}
	Then the sample variance ${S}^{2}(L_1,L_2)$ in these $N_b$ estimates is used to estimate 
	the statistical error bars in the final estimate,
	\begin{equation}
		{\sigma}_{d_{s}}(L_1,L_2) = \frac{{S}(L_1,L_2)}{\sqrt{N_b - 1}} \ .
	\end{equation}
	At this point we have obtained a single estimate of $d_{s}$ 
	(with error bars) for each consecutive pair of systems simulated. 
	
	Next we look to see if this $\overline{d_{s}(L_1,L_2)}$ converges to some 
	value as $L_1$ and $L_2$ tend to infinity.  We extrapolate 
	the infinite system size limit using a variety of least squares fitting 
	routines adapted from the GNU Scientific Library \cite{GSL}.  
	Standard scaling suggests that $d_{s}$ as a function of system size
	$L$ will exhibit power law scaling behavior.  But since we have no knowledge
	of the expected value for the exponent in this power law, we can
	allow this exponent to vary as a parameter in our fit, plotting
	$\overline{d_{s}}$ vs.~$L^{-\lambda}$, where we take $L$ to be the geometric
	mean $L = \sqrt{L_1 L_2}$.
	Allowing $\lambda$ to vary, fitting the data gives an estimate
	for $d_s$ as well as the correction to scaling exponent $\lambda$.  
	See Fig.~\ref{fig:linear-fit-extrapolation} for an example of
	one extrapolation attempt, where $\lambda$ is seen to be roughly $0.5$ and
	linear least squares fitting for $\overline{d_{s}}$ vs.~$L^{-\lambda}$ is used
	for the four largest system sizes in $d=2$ dimensions.

\begin{figure}
	\includegraphics[%
		width=\singfigwidth, 
		keepaspectratio]{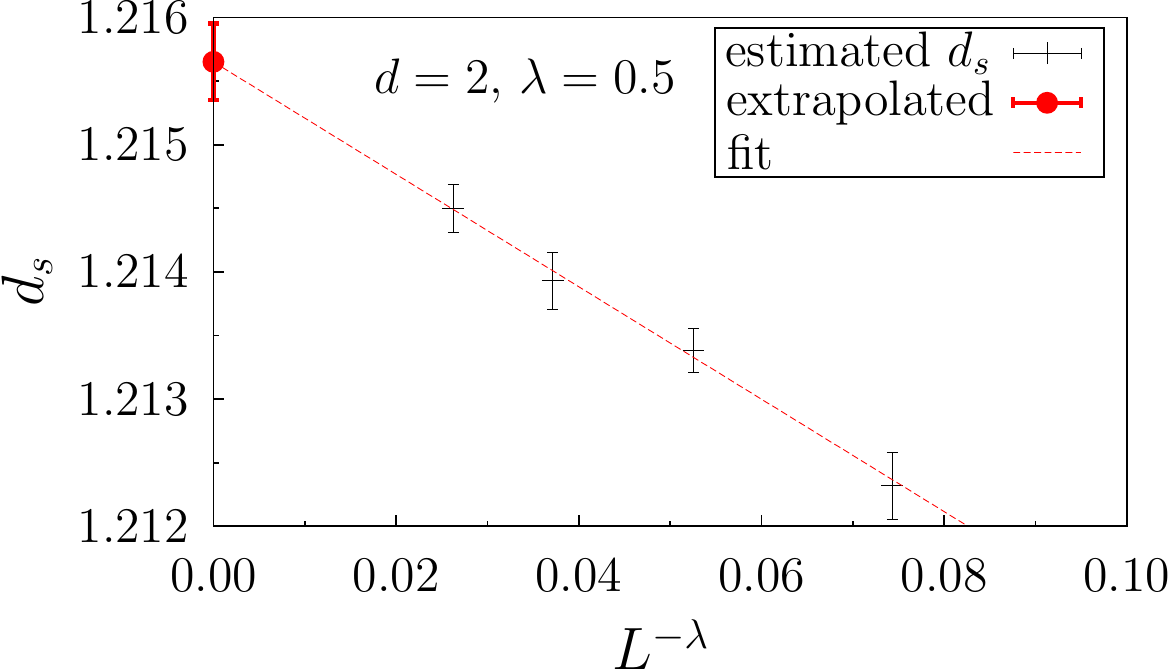} 
	\caption{\label{fig:linear-fit-extrapolation} An illustration
		of a linear least squares fitting method being used to extract
		the value of $d_{s}$ in the infinite system size limit in 
		dimension $d=2$. The fit uses the form $d_{s}(L) = A L^{-\lambda} + d_{s}(\infty)$ 
		where $\lambda$ is a correction to scaling exponent and $A$ and $d_{s}(\infty)$ are fitting 
		parameters.  Here $\lambda = 0.5$, and $L = \sqrt{L_1 L_2}$, where 
		$L_{1}$ and $L_{2}$ are the two system sizes used to produce a 
		given data point.  The fit found gives $d_{s}(\infty) = 1.216(1)$.
	}
\end{figure}

\subsection{Blind test for analysis method}

	To test this data analysis procedure, we applied the procedure
	to the similar problem of the uniform spanning tree (UST) in dimension
	$d=2$ \cite{Read}.
	Whereas a minimal spanning tree seeks to minimize the total
	weight of a tree that spans all the vertices of the lattice, a
	uniform spanning tree is simply any tree that spans the vertices 
	of the lattice, chosen with equal weight from all possible spanning
	trees.  It can be thought of as a generalization of the
	minimal spanning tree problem to a system where all edges have
	the exact same weight. The reason for using this system as a
	test case for the analysis method is that one can examine $r(s)$ 
	data to look at $d_{s}$ for paths on the UST, just as one
	would examine such paths on an MST. The analysis should be
	completely analogous, and we observed directly that the 
	correlations in the $r(s)$ data for the UST have a 
	similar structure to that of the MST, allowing us to use 
	the form of Eq.~(\ref{eqn:correlation-length-model}) for $\ell$.
	To ensure no bias in this test, one of us was not made aware of the
	nature of the system at the time of this test but was only 
	given the prepared $r(s)$ data on which to blindly run the 
	analysis. The final result from this blind test, using a range of
	systems of size $128^2$ to $1024^2$, was
	$d_{s} = 1.2499(4)$ for the UST, which agrees well with 
	the known exact result $d_{s} = {5}/{4}$ in two dimensions \cite{Majumdar}.
	This provides confidence in the data analysis method.

\section{Results}

	Table \ref{table:ds} displays our final numerical estimates for
	$d_{s}$. The values in dimensions $d=2$ and 
	$d=3$ agree with previous results
	\cite{Cieplak-Maritan-Banavar,Buldyrev-et-al,Sheppard-et-al, Wieland-Wilson}
	and have error bars that are similar or smaller.  Our result for $d_{s}$ in 
	dimension $d=4$ is a bit higher than the result 1.59(2) by Cieplak, Maritan, 
	and Banavar \cite{Cieplak-Maritan-Banavar}.
\setlength{\tabcolsep}{0.3cm}
\begin{table}%
	\caption{$d_{s}$ calculated using MST algorithms} 
	\centering 
	\begin{tabular}{c c c} 
		\hline\hline                        
		Dimension $d$ & Two-step & Prim's (no trimming) \\ [0.5ex] 
		\hline                  
		2 & 1.216(1) & 1.216(1) \\
		3 & 1.46(1) & 1.46(1)  \\
		4 & 1.65(2) & 1.66(2) \\
		5 & 1.86(4) & 1.86(4) \\ [1ex]      
		\hline\hline 
		\end{tabular}
		\label{table:ds} 
\end{table}
	The fractal dimensions computed from the two-step MTISC 
	agree with those of the intermediate, untrimmed Prim's MTISC, 
	to within our error estimates.

	Our results for the path length dimension $d_{s}$ can be directly compared
	with the $\mathcal{O}(\epsilon = 6-d)$ expansion of Jackson and Read.  A
	graphical comparison of the data is shown in 
	Fig.~\ref{fig:results-vs-predictions}.  Our results are not in
	conflict with the first order perturbation theory calculations.
	Some previous comparisons of $\mathcal{O}(\epsilon)$ calculations
	with numerical results for disordered materials show differences of similar magnitude
	\cite{Dahmen-Sethna}.

	We investigate possible $\mathcal{O}({\epsilon}^2)$ calculations using 
	nonlinear fitting routines adapted from the GNU
	Scientific Library \cite{GSL} to fit $d_{s}$ vs.~$\epsilon$.
	Using a two parameter fit,
	\begin{equation} 
		\label{eqn:2nd-order-two-parameter-fit}
		d_{s} = 2 + a\epsilon + b{\epsilon}^2 \ ,
	\end{equation}
	and allowing $a$,$b$ to vary, we find a chi-squared of $3.8$ for two d.o.f.,
	suggesting consistency with a quadratic fit to within our errors.
	For this fit, we find $a = -0.142(8)$, which is near to the $-1/7$ suggested by
	Jackson and Read.  We find for this fit the value $b = -0.014(2)$ for the second order
	prefactor.

	Fixing $a=-1/7$, a one parameter fit,
	\begin{equation}
		\label{eqn:2nd-order-one-parameter-fit}
		d_{s} = 2 - \frac{\epsilon}{7} + b{\epsilon}^2 \ ,
	\end{equation}
	gives $b=-0.0133(1)$ with a chi-squared of $3.8$ for three d.o.f.  
	Presuming that the Jackson and Read result is correct to first order, 
	this gives us a more precise numerical prediction for the second order term.  
	This fit, Eq.~(\ref{eqn:2nd-order-one-parameter-fit}), is plotted in Fig.~\ref{fig:results-vs-predictions}.

\begin{figure}
	\includegraphics[%
		width=\singfigwidth, 
		keepaspectratio]{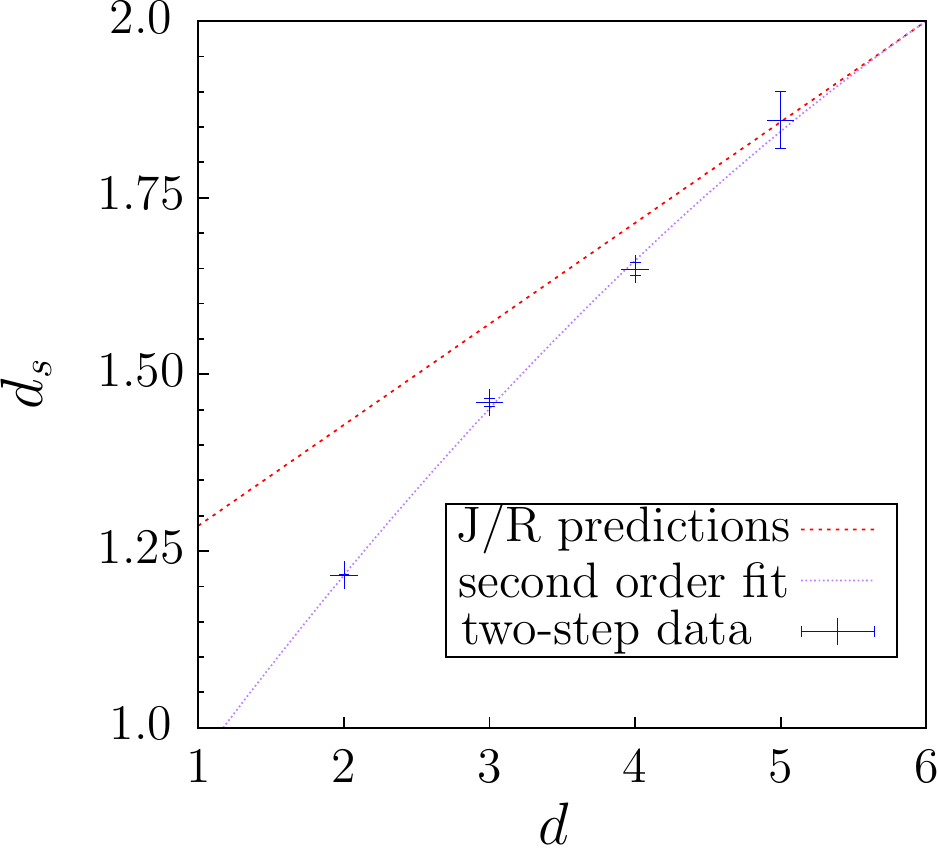} 
	\caption{\label{fig:results-vs-predictions} A plot comparing 
		numerical results for $d_{s}$ using the two-step method
		against predictions from the $\mathcal{O}(\epsilon)$ perturbation expansion 
		theorized by Jackson and Read.  Also included is an example
		of a compatible $\mathcal{O}({\epsilon}^2)$ fit, 
		$d_{s} = 2 - \frac{\epsilon}{7} +b{\epsilon}^{2}$, with a best fit value 
		$b=-0.0133$.
	}
\end{figure}

\section{Summary} 

	The intention of the two-step method was to allow the simulation
	of larger systems than were previously possible, reducing
	memory requirements of MST-finding algorithms by combining
	Prim's algorithm with Kruskal's algorithm.  In this regard
	the work was successful, allowing for precise calculations
	of $d_{s}$.  The trimming of the Prim's algorithm tree
	to possibly improve scaling of MTISCs constructed appears to have been
	unnecessary.  Calculations using 
	Prim's algorithm alone yielded almost identical results to those
	that used the two-step method.

	We developed a data analysis method that allowed taking nonuniform
	correlations into account in order to obtain more accurate
	estimates for $d_{s}$.  This analysis method should be
	applicable to a wide array of disordered systems with scale invariance 
	and may be useful for future work.

	The results for $d_{s}$ calculated in this work are
	compatible with the perturbation expansion result
	proposed by Jackson and Read.  Fitting $d_{s} = 2 + a\epsilon + b{\epsilon}^2$,
	we find $a = -0.142(8)$, compatible with the predicted $a=-1/7$ \cite{Jackson-Read-II}.
	Fixing the first order result to the Jackson and Read result,
	we used an $\mathcal{O}({\epsilon}^2)$ 
	fit, $d_{s} = 2 - \frac{\epsilon}{7} +b{\epsilon}^{2}$, yielding
	$b=-0.0133(1)$.  This could be checked if a higher order analytic 
	calculation could be computed.

	This work was made possible in part by NSF Grant No.~DMR-1006731 and 
	by the Syracuse University Gravitation and Relativity computing cluster, 
	which is supported in part by NSF Grant No.~PHY-0600953.
	This work was carried out largely using the Syracuse University HTC
	Campus Grid, a computing resource of approximately 2000 desktop 
	computers supported by Syracuse University.
 	Some of this work was discussed at the Aspen Center for Physics
	(NSF Grant No.~1066293).

\section{Appendix}

In this appendix, we present more precise definitions of the algorithms
used to generate MTISCs, as well as details of the connection between $T_{P}$, $T_{K}$, 
and $T_{2}$.   We also discuss the region in $s$ used for the ${\chi}^{2}$ fitting 
procedure to estimate the fractal dimension $d_{s}$.

\subsection{Definitions and Algorithms}{\label{appendix:definitions}}
\numberwithin{equation}{subsection} 

\newenvironment{definition}[1][Definition]{\begin{trivlist}
\item[\hskip \labelsep {\bfseries #1}]}{\end{trivlist}}

\newenvironment{algorithm}[2][Algorithm:]{\begin{trivlist}
\item[\hskip \labelsep {\bfseries #2 #1}] \item}{\end{trivlist}}

\newenvironment{proof}[2][Proof]{\begin{trivlist}
\item[\hskip \labelsep {\bfseries #1}] {#2}}{\end{trivlist}}

\newcounter{Lemmacount}
\setcounter{Lemmacount}{0}
\newenvironment{lemma}[1][Lemma \arabic{Lemmacount}]{\stepcounter{Lemmacount} \begin{trivlist}
\item[\hskip \labelsep {\bfseries #1}]}{\end{trivlist}}

\newenvironment{corollary}[1][Corollary \arabic{Lemmacount}]{\stepcounter{Lemmacount} \begin{trivlist}
\item[\hskip \labelsep {\bfseries #1}]}{\end{trivlist}}

\newcounter{Casecount}
\setcounter{Casecount}{0}
\newenvironment{proofcase}[1][Case \arabic{Casecount}]{\stepcounter{Casecount} \begin{trivlist}
\item[\hskip \labelsep {\bfseries #1}]}{\end{trivlist}}

Here we present a more detailed discussion of Kruskal's and
Prim's algorithms.  We consider these MST-finding algorithms in the context
of an undirected weighted graph $G=(V,E,w)$, where $V$ is a set of 
vertices, $E$ is a set of edges connecting these vertices, and 
we define a weight function $w:E \mapsto \mathbb{R}$, with each edge 
$e \in E$ having weight $w(e)$.  We consider the case of unique weights; 
no two edges in $E$ have exactly the same weight.

In order to precisely describe the algorithms, it will be helpful to define
some terms.  A cycle or loop is a closed path such that the removal of any single 
edge from this path will result in the path becoming an open path,
a connected set of edges with no vertex shared by more than two edges. 
When constructing trees, the notion of allowed and forbidden cycles
is useful.  Every cycle in the finite set $\mathcal{C}$, the set of all 
possible cycles for graph $G$, is chosen to belong to $\mathcal{C}_F$ (the set 
of forbidden cycles) or $\mathcal{C}_A$ (the set of allowed cycles), where 
$\mathcal{C}_F \cup \mathcal{C}_A = \mathcal{C}$ and 
$\mathcal{C}_F \cap \mathcal{C}_A = \emptyset$.  In a typical application 
on a periodic lattice, the allowed cycles correspond to loops that wrap 
around the lattice, while forbidden cycles correspond to non-wrapping 
loops that can be deformed, plaquette by plaquette, to a point.  
A set of two or more edges is considered connected if every edge in 
the set shares a vertex with at least one other edge in the set.
A cluster $T$ is set of connected edges and the 
vertices that these edges connect. A cluster may contain 
cycles, whereas a tree is an acyclic cluster. $T_E$ denotes 
the set of edges in the cluster and $T_V$ denotes the set of vertices 
in the cluster.

Consider edge $e=(u,v)$ where $e \in E$ and $u,v \in V$.  Adding $e$ 
to a cluster $T$ means that the edge set for the cluster, $T_E$, becomes 
$T_E \cup \{e\}$, and the set of vertices in the cluster, $T_V$, becomes $T_V \cup \{u,v\}$.
A forbidden edge for a cluster $T$ is an edge that, if added
to $T$, would cause $T$ to gain a forbidden cycle.
An edge that is not a forbidden edge is an allowed edge.
An allowed terminating edge for a cluster $T$ is an allowed edge
that, if added to $T$, would cause $T$ to gain an allowed cycle
(a wrapping loop).  The addition of this allowed cycle will be used as 
the termination condition for both Kruskal's and Prim's algorithms, which 
are described below.  $\partial T$ is the set of adjacent edges for a cluster
$T$, consisting of all edges that have one (but not both) endpoints in cluster $T$.
As this set of edges forms a border or frontier on the outer regions of the cluster,
we call $\partial T$ the frontier of cluster $T$.
The forbidden frontier of this cluster, $\partial_{F} T$, is 
the subset of the adjacent edges that are forbidden edges for $T$.
The allowed frontier of this cluster, $\partial_{A} T$, 
is the subset of the adjacent edges that are allowed edges for $T$.

To examine an edge in Kruskal's or Prim's algorithm is 
to select this edge during a step of the algorithm and 
decide whether to accept or reject this edge into a cluster based 
on the conditions of the algorithm (usually dealing with the weight 
of the edge and whether this edge is allowed or forbidden for a 
particular cluster).  

Next we will present Kruskal's and Prim's 
algorithms, step by step, using the notation we have outlined above.

\begin{algorithm}{Kruskal's}
\begin{enumerate}
\item Sort $E$ by increasing weight $w$ to form a list of edges $L$.
\item Initialize each vertex in the graph to be its own tree, 
			not connected to any other vertices and containing no edges.
\item Select the first (lowest weight) edge $e=(u,v) \in L$.
			Remove this edge from $L$.  If $e$ is a allowed edge, join 
			into a single tree the tree containing $u$ with the tree 
			containing $v$, adding edge $e$.  If $e$ is a forbidden edge 
			(that is, adding edge $e$ to to the tree(s) containing its 
			endpoints $u$ and $v$ would introduce a forbidden cycle), 
			edge $e$ is disregarded and not added to any trees.  
			Edge $e$ has now been examined.
\item Repeat step (3) until a allowed terminating edge, i.e., one that introduces
			an allowed (wrapping) cycle, is selected.  The 
			tree that contains both the vertices of this edge
			is identified as the Kruskal's MTISC $T_{K}$, 
			and the allowed terminating edge that is examined in this final
			step is the Kruskal's wrapping edge, $e_{K} \in E$.
\end{enumerate}
\end{algorithm}

\begin{algorithm}{Prim's}
\begin{enumerate}
\item Initialize the growing Prim's tree $T_{g}$ to have one site, 
			called the Prim's origin, $v_0 \in V$.  Note that $T_{g}$
			is both a function of the $v_0$ and time (number of iterations of 
			Prim's algorithm), since $T_{g}$ ``grows'' from the origin $v_0$
			by adding edges and vertices as Prim's algorithm proceeds.
\item Sort the frontier of the current Prim's tree, $\partial T_{g}$, by increasing 
			edge weight $w$ to form the Prim's queue (a function of $T_{g}$),
			excluding any edges that have already been examined.
\item Select the first (lowest weight) edge $e=(u,v)$ in 
			the Prim's queue.  
			If $e$ is a allowed edge, add it to $T_{g}$, and if either $u$ or $v$ 
			is not already in $T_{g}$, add it to $T_{g}$ as well.
			Otherwise, if $e$ is a forbidden edge, disregard the edge and do not add it to $T_{g}$.
  		Edge $e$ has now been examined.
\item Repeat steps (2)-(3) until a allowed terminating edge, i.e., one that 
			introduces an allowed (wrapping) cycle, is selected.  
			At this time, the growing Prim's tree $T_{g}$ is identified
			as the Prim's MTISC $T_{P}$, and the allowed terminating edge examined in this
			final step is the Prim's wrapping edge, $e_{P}(v_{0}) \in E$.  
			Note that $T_{P}$ and $e_{P}$ are both functions of 
			the Prim's origin $v_0$.
\end{enumerate}
\end{algorithm}

\subsection{Proof of validity for two-step method}{\label{appendix:proof}}

The two-step method first applies Prim's algorithm
to find $T_{P} \cup e_{P}$.  Kruskal's algorithm is then 
executed using the edges in $T_{P} \cup e_{P}$, serving as a trimming
procedure for $T_{P}$ and yielding what we 
will term the two-step MTISC, $T_{2}$.  Here we will prove that the 
two-step method for MTISC generation yields $T_{2} = T_{K}$ in 
all cases except in the case of multiple disjoint ISCs.  All results 
of this method can be divided into three cases.  The first is the 
case where $T_{P}$ is exactly equal to $T_{K}$ and no trimming is needed, 
yielding $T_{P} = T_{2} = T_{K}$.  In the second case, $T_{K}$ is a subset 
of $T_{P}$, and the trimming procedure of the two-step method yields $T_{2} = T_{K}$.  
In the third case, there are multiple disjoint ISCs, and $T_{P}$ and $T_{K}$ share no 
common edges or vertices, meaning that $T_{2} \ne T_{K}$.  In other words, we will 
prove that 
\begin{equation}
	(T_{P} \supseteq T_{K}) \vee (T_{P} \cap T_{K} = \emptyset)
\end{equation}
and subsequently that $T_{2} = T_{K}$ in all cases except when 
the Prim's MTISC and the Kruskal's MTISC do not intersect.  Here it will be useful 
to establish a few lemmas to aid in verifying these three cases we have outlined.

\begin{lemma}
Edges in the allowed frontier of $T_{K}$ have higher weight than edges
in $T_{K}$.
\end{lemma}

\begin{proof}
Because Kruskal's algorithm examines edges that are monotonically 
increasing in weight as the algorithm proceeds, at the time
Kruskal's algorithm terminates, the edges that are not examined must 
be higher weight than edges that have been examined.  Since edges 
in the allowed frontier $\partial_{A} T_{K}$ cannot be in $T_{K}$, this means that edges in 
$T_{K}$ have been examined at the time of Kruskal's algorithm termination, 
whereas edges in $\partial_{A} T_{K}$ have not. Thus edges in $\partial_{A} T_{K}$ 
must be higher weight than edges in $T_{K}$.
\end{proof}

\begin{lemma}
For any given cycle $c \in \mathcal{C}$ for which all edges of $c$ are examined 
in Kruskal's algorithm, the edge in this cycle that is examined last must be 
the highest weight edge in this cycle.
\end{lemma}

\begin{proof}
Because Kruskal's algorithm examines edges that are monotonically 
increasing in weight as the algorithm proceeds, for the finite set of 
edges in $c$, the highest weight edge in this set will be examined 
last in Kruskal's algorithm.
\end{proof}

\begin{corollary}
As a corollary, the Kruskal's wrapping edge $e_{K} \in E$ has higher weight 
than any edge in the Kruskal's MTISC $T_{K}$.
\end{corollary}

\begin{proof}
Because $e_{K}$ is the last edge examined during the running 
of Kruskal's algorithm, it must have a weight higher than every other edge that 
is examined during the running of the algorithm, including every edge in $T_{K}$.
\end{proof}

\begin{lemma} 
From lemma 2, for any edge $e$ in the forbidden frontier of $T_{K}$, where adding 
$e$ to $T_{K}$ would form a forbidden cycle $c_e$, $e$ must be the highest 
weight edge in the forbidden cycle $c_e$.
\end{lemma}

\begin{proof}
The criteria for an edge $e$ to be in the forbidden frontier $\partial_{F} T_{K}$ is
that adding $e$ to $T_{K}$ would form a forbidden cycle $c_e$.  This
means that every edge in $c_{e}\backslash\{e\}$ is in $T_{K}$ and is
examined before edge $e$ during Kruskal's algorithm.  Thus, edge $e$ will
have higher weight than any other edge in cycle $c_e$.
\end{proof}

Using these lemmas that we have established, we will show that each
realization of the two-step method can be categorized into one of three distinct 
cases, with cases 1 and 2 yielding $T_{2} = T_{K}$.  In the third case, we see
that there are no common edges or vertices between $T_{P}$ and $T_{K}$, so
$T_{2} \ne T_{K}$ in this case.

\begin{proofcase}
If the Prim's origin $v_0$ is a vertex in the Kruskal's MTISC $T_{K}$, then the Prim's
MTISC $T_{P}$ will be identical to the Kruskal's MTISC; $T_{P} = T_{K}$. 
\end{proofcase}

\begin{proof}
\item Since Prim's algorithm grows a tree $T_{g}$ by examining and adding edges
adjacent to $T_{g}$, it will be crucial for us to pay close attention to the frontier of
$T_{g}$, $\partial T_{g}$.  This frontier is sorted to form the Prim's queue from which 
edges are selected and considered for addition to $T_{g}$.  In this first case, as long as 
we are only adding to $T_{g}$ vertices that are also in $T_{K}$, the Prim's queue will consist 
entirely of edges that are either in $T_{K}$ or $\partial T_{K}$. The fulfillment of this 
condition,
\begin{equation}
	\label{case1-condition}
	\partial T_{g} \subseteq T_{K} \cup \partial T_{K} \ ,
\end{equation}
will be shown in the remainder of the proof.  

\item From lemma 1, when the frontier of the Prim's tree is sorted in Prim's algorithm 
to form the Prim's queue, any edges that are also in $T_{K}$ will be 
placed before (having lower weight than) those edges in the allowed 
frontier of $T_{K}$.  This means that all edges in $T_{K}$ will be examined and added 
to $T_{g}$ earlier in Prim's algorithm than edges in the allowed frontier of $T_{K}$.
As the Prim's origin $v_0$ is in $T_{K}$, we are guaranteed to
have at least one edge from $T_{K}$ added to the Prim's queue at the
start of the Prim's growth.  Furthermore, since $T_{K}$ is connected, as edges 
from $T_{K}$ are added to $T_{g}$ through the Prim's growth, more edges from 
$T_{K}$ and its frontier $\partial T_{K}$ will be added to the Prim's queue,
ensuring that $T_{g}$ has more edges from $T_{K}$ to add during further steps of Prim's
algorithm.  $T_{g}$ begins as a single vertex $v_0$ and grows to include more and
more edges and vertices from $T_{K}$ as Prim's algorithm proceeds.

\item Though we have shown that edges in $T_{K}$ will be examined during
Prim's algorithm before any edges in the allowed frontier of $T_{K}$, we 
cannot say the same about edges in its forbidden frontier, $\partial_{F} T_{K}$, which may 
also be in the Prim's queue and in consideration for selection during the Prim's 
growth.  If edge $x=(a,b) \in \partial_{F} T_{K}$ is in the Prim's queue, from lemma 4 we
know that both vertices $a$ and $b$ are in $T_{K}$.  If we call $c_x$ the forbidden cycle 
that would be created in $T_{K}$ by adding edge $x$ to $T_{K}$, we also know from lemma 4 
that $x$ has higher weight than every other edge in the forbidden cycle $c_x$. 

\item Thus, all edges in $c_{x}\backslash\{x\}$ will be examined in
Prim's algorithm and added to $T_{g}$ before $x$ is examined.  So 
when $x$ is finally examined in Prim's algorithm, $x$ will be a forbidden edge 
for $T_{g}$ and will not be added to $T_{g}$.  
This ensures that any forbidden cycles for $T_{K}$ will be handled in 
the same manner in both Kruskal's and Prim's algorithms; no edges in the forbidden 
frontier of $T_{K}$ will be added to $T_{g}$.

\item Using corollary 3, we can see that until the Kruskal's wrapping edge $e_{K}$ is 
examined in Prim's algorithm, Eq.~(\ref{case1-condition}) will remain satisfied and all of 
the edges and vertices in $T_{K}$ will be added to $T_{g}$ before $e_{K}$ is 
examined.  Furthermore, as we have shown, none of the edges (allowed or forbidden) in the frontier 
of $T_{K}$ will be added to the growing Prim's tree $T_{g}$ before the Kruskal's wrapping edge $e_{K}$ 
is examined.  In this case, $e_{K}$ will also serve as the Prim's wrapping edge $e_{P}$.  Since the 
condition~(\ref{case1-condition}) is satisfied at the time $e_{K} = e_{P}$ is examined 
in Prim's algorithm (terminating the algorithm by adding an allowed cycle), $T_{P} = T_{K}$. 
In this case, the Kruskal's trimming step of the two-step method is unnecessary,
and the two-step method will yield $T_{2} = T_{K}$.
\end{proof}

\begin{proofcase}
If the Prim's origin $v_0$ is not in the Kruskal's MTISC $T_{K}$ but the Prim's tree $T_{g}$ 
``grows'' to include any vertices of $T_{K}$, then the Prim's MTISC $T_{P}$ will include all
of the Kruskal's MTISC $T_{K}$, i.e., $T_{P} \supset T_{K}$.
\end{proofcase}

\begin{proof}
\item If we start Prim's algorithm from an origin $v_0$ that is not a vertex 
in the Kruskal's MTISC $T_{K}$, we can define bridge edge $b=(u_b,v_b)$ as the first edge added 
to the growing Prim's tree $T_{g}$ that introduces a vertex of $T_{K}$ into $T_{g}$.  
So $b=(u_b,v_b)$ is the first edge added to $T_{g}$ for which either 
$u_b$ or $v_b$ are in $T_{K}$.  

\item At the point in Prim's algorithm when $b$ is added to $T_{g}$, 
$b$ has a weight lower than any edge in the Prim's queue.
Otherwise, edge $b$ would not have been selected for addition to the Prim's tree at
that time.

\item Examining edge $b$ in the frame of the Kruskal's MTISC $T_{K}$, we know that edge $b$ 
is in the allowed frontier of $T_{K}$ since only one of $u_b$ or $v_b$ is in $T_{K}$.  
This means that $b$ cannot be in $T_{K}$ itself, nor can $b$ be in the forbidden frontier 
of $T_{K}$.  Because $b$ is in the allowed frontier of $T_{K}$, it has a weight higher than
any edge in $T_{K}$, by lemma 1.  Thus, Prim's algorithm will continue as in case 1, 
with growth continuing from this first vertex of $T_{K}$ (either $u_b$ or $v_b$)
that is introduced to $T_{g}$.  Any edges from $T_{K}$ that are added to the Prim's 
queue will be added to the front (lower weight end) since these edges will all have a 
weight lower than $w(b)$, whereas every edge in the Prim's queue thus far has 
weight higher than $w(b)$.

\item Prim's algorithm will add edges from the Kruskal's MTISC $T_{K}$ until reaching termination 
with the examination of the Kruskal's wrapping edge $e_{K}$, which in this case will also be the
Prim's wrapping edge $e_{P}$.  At this point $T_{P} \supset T_{K}$.  In this case, the Kruskal's algorithm
portion of the two-step method will serve to trim from the Prim's tree the region near the origin $v_0$ up to
the bridge $b$, leaving $T_{2} = T_{K}$.
\end{proof}

\begin{proofcase}
If the Prim's origin $v_0$ is not in the Kruskal's MTISC $T_{K}$ and the growing Prim's tree $T_{g}$ 
does not ``grow'' to include any vertices of $T_{K}$ before Prim's algorithm terminates, then
the Prim's MTISC and the Kruskal's MTISC will not intersect.  In other words,
$T_{P} \cap T_{K} = \emptyset$.
\end{proofcase}

\begin{proof}
\item In this case, during Prim's algorithm, the Prim's wrapping edge $e_{P}$ is examined before 
any vertices of the Kruskal's MTISC $T_{K}$ are examined or added to the Prim's tree.  The 
algorithm terminates with $T_{P} \cap T_{K} = \emptyset$.  In this case, the two-step method
will yield $T_{2} \ne T_{K}$.  However, this case is uncommon, and it is seen by direct observation
that $T_{2}$ and $T_{K}$ have similar scaling properties.
\end{proof}

Taking all three cases for the two-step method, we can say that either the Prim's MTISC $T_{P}$ contains
or is equal to the Kruskal's MTISC $T_{K}$, or the Prim's and Kruskal's MTISCs do not intersect.
Symbolically, $(T_{P} \supseteq T_{K}) \vee (T_{P} \cap T_{K} = \emptyset)$.
In cases 1 and 2, where $T_{P}$ and $T_{K}$ do overlap, the two-step MTISC $T_{2}$ will be equal
to the Kruskal's MTISC $T_{K}$.  So in these cases the two-step method yields the same result 
as that obtained from running Kruskal's algorithm on the entire graph $G$.
We see that $T_{2} = T_{K}$ in all cases except the case of multiple disjoint ISCs, 
where the Prim's MTISC $T_{P}$ (and by construction the two-step MTISC $T_{2}$ as well) does not 
intersect with the Kruskal's MTISC $T_{K}$.

\subsection{Fitting region for ${\chi}^{2}$}{\label{appendix:fitting_region}}

	Here we will discuss the data cutoffs we impose in order to restrict exactly what region in $s$ 
	of data we want to fit.  These data cutoffs allow us to reduce overfitting errors in our ${\chi}^{2}$ 
	fitting routine (Eq.~(\ref{eqn:custom-chisquared-definition})) for estimating the fractal dimension $d_{s}$.  
	We consider a small $s$ (lower) cutoff by implementing $s_{l}$ as the minimum $s$ value allowed 
	into the fitted data.  We also consider a large $s$ (upper) data cutoff by enforcing an 
	$\omega_{u}$ as the maximum allowed value for any of the scaled data in the fit.  Introducing an 
	$\omega_{u}$ cutoff seemed natural because $\omega \gg 1$	corresponds to 
	paths that are much longer than $L^{d_{s}}$ and are rare.
	Since these long paths are less frequent, data points with a large $\omega$ value
	have higher statistical uncertainties than those with smaller $\omega$ values.  
	For small $s$ we consider lattice effects. It made sense to use $s_{l}$ as a low cutoff 
	since we have small discrete $s$ values (steps) in the paths.  Using $\omega_{l}$ 
	would also have been a viable option, but it is easier to interpret the physical
	meaning of $s_{l}$, since one unit of $s$ corresponds to the lattice spacing for 
	all system sizes.

	To determine reasonable values to use for the upper and lower data cutoffs,
	we tested our ${\chi}^{2}$ fitting routine with various values of these cutoffs and 
	assessed how well the minimum value of ${\chi}^{2}$ agreed with a predicted 
	estimate ${\chi}_{p}^{2}$.  To estimate ${\chi}_{p}^{2}$, it is instructive to envision
	comparing data sets from two different system sizes on scaled axes $\rho$ vs.~$\omega$, 
	as in Fig.~\ref{fig:single-pair}(a).  We use a discrete set $\{\omega_{k}^0\}$, consisting of 
	$n$ uniformly spaced values of $\omega$, to calculate a set of $\rho(\omega_{k}^0)$ values for each of 
	the two data sets.  The comparison between these $\rho$ values goes into the calculation of 
	${\chi}_{2}$, as outlined in Eq.~(\ref{eqn:custom-chisquared-definition}).  
	
	One can imagine breaking the $\omega$ axis up into $b$ segments or ``boxes'' of length equal to 
	the scaled correlation length $\ell$.  Each box will contain some number $n_b$ of data points from the 
	set of $n$ points in the $\{\omega_{k}^0\}$ discrete points we use to calculate the goodness of the collapse.  
	In this way, we can estimate ${\chi}^{2}_{p}$ piece by piece, calculating separately the contribution
	${\chi}_{b}^{2}$ from each of the $b$ boxes:  
	\begin{equation}
		{\chi}_{p}^{2} = {\chi}_{b}^{2} b \ .
	\end{equation}
	Further, we can estimate $n_b$, the number of data points that fall in each of these boxes.  
	This allows us to subdivide the ${\chi}_{b}^{2}$ into contributions from each individual data point, ${\chi}_{1}^{2}$:
	\begin{equation}
		\label{eqn:chisquared-predicted-partition}
		{\chi}_{p}^{2} = {\chi}_{1}^{2} n_b b \ .
	\end{equation}

	Now that we have partitioned ${\chi}_{p}^{2}$ by this relation, we can calculate
	${\chi}_{1}^{2}$, $n_b$, and $b$ individually.  First, we can estimate ${\chi}_{1}^{2}$, the ${\chi}^{2}$ 
	contribution from a single data point in the set of $\{\omega_{k}^0\}$ values, by its expectation value 
	$E({\chi}_{1}^{2})$.  We can write, using the form of Eq.~(\ref{eqn:custom-chisquared-definition}),
	\begin{equation}
		\label{eqn:chisquared-predicted-single-point}
		\begin{split}
			 {\chi}_{1}^{2} &\approx E({\chi}_{1}^{2}) \\
			 &\approx \frac{E([{\rho}_{1} - {\rho}_{2}]^{2})}{\ell [{\sigma}_{1}^{2} + {\sigma}_{2}^{2}]} \\
												 &\approx \frac{E({\rho}_{1}^{2}) + E({\rho}_{2}^{2})}{\ell [{\sigma}_{1}^{2} + {\sigma}_{2}^{2}]} \\
												 &\approx \frac{{\sigma}_{1}^{2} + {\sigma}_{2}^{2}}{\ell [{\sigma}_{1}^{2} + {\sigma}_{2}^{2}]} \\
												 &\approx \frac{1}{\ell} \ ,
		\end{split}
	\end{equation}
	where $\rho_1$($\rho_2$) is the $\rho$ value taken from data set 1(2), and ${\sigma}_{1}^{2}$(${\sigma}_{2}^{2}$) is
	the variance at this point for data set 1(2).

	Next we can write an expression for the number of data points per box, $n_b$, by multiplying the length
	of each box, $\ell$, by the density of data points,
	\begin{equation}
		\label{eqn:nb}		
		n_b = \ell \left( \frac{n}{\omega_{u}-\omega_{l}} \right) \ ,
	\end{equation}
	where again $n$ is the total number of data points in the set $\{\omega_{k}^0\}$ being used for the ${\chi}^{2}$
	calculation and ${\omega_{u}-\omega_{l}}$ gives the full allowed range of ${\omega}_{k}^0$ values.

	Finally, we compute the total number $b$ of boxes, $b$, by dividing the
	range ${\omega_{u}-\omega_{l}}$ into two regions, $\omega < 1$
	and $\omega > 1$.  This allows us to separately compute the number of boxes in
	each of these regions, $b_{<}$ and $b_{>}$ respectively.  Of course, $b = b_{<} + b_{>}$.
	Expressing $b_{>}$ is trivial, since in this region the scaled correlation length is constant 
	at $\ell = A$ (as given by our model in Eq.~(\ref{eqn:correlation-length-model})),
	where $A$ is the proportionality constant relating $\ell$ and $\omega$ by $\ell = A \omega$.
	Thus we write
	\begin{equation}
		\label{eqn:b-greater}
		b_{>} = \frac{{\omega}_{u} - 1}{A} \ .
	\end{equation}
	Writing $b_{<}$ is not quite as simple. It may help to
	begin at $\omega = 1$ and think of stepping left toward lower values of $\omega$ and
	eventually to ${\omega}_{l}$, counting up the number of boxes we step across.  The boundary 
	of the first box will occur at $\omega = (1-A)$, since the scaled correlation length $\ell$ 
	is equal to $A$ around $\omega = 1$.  Likewise, the edge of the second box we approach will 
	fall at $\omega = (1-A)^{2}$, the $i^{\text{th}}$ box will reach to 
	$\omega = (1-A)^{i}$, and the final ${b}_{<}^{\text{th}}$ box will fall 
	at the lowest allowed $\omega$ value, ${\omega}_{l}$.  So we can write
	\begin{equation}
			 {\omega}_{l} = \left(1-A\right)^{b_{<}} \ ,
	\end{equation}
	from which is follows that
 	\begin{equation}
		\begin{split}
		\label{eqn:b-less}
			b_{<}	 &= \frac{\ln\left({\omega}_{l}\right)}{\ln\left(1-A\right)} \\
						 &\approx \frac{\ln\left({1}/{{\omega}_{l}}\right)}{A} \ .
		\end{split}
	\end{equation}
	Here an approximation is made in the denominator since $A$ is small.

	Combining Eqs.~(\ref{eqn:chisquared-predicted-partition}), (\ref{eqn:chisquared-predicted-single-point}), 
	(\ref{eqn:nb}), (\ref{eqn:b-greater}), and (\ref{eqn:b-less}), we can write an expression for 
	${\chi}_{p}^{2}$ by adding the contributions from both of the two regions we examined, 
	$\omega < 1$ and $\omega > 1$:
	\begin{equation}
		\begin{split}
			{\chi}_{p}^{2} &= \frac{1}{\ell} \left[\ell \left(\frac{n}{\omega_{u}-\omega_{l}}\right)\right] \left[\frac{\ln\left({1}/{\omega_{l}}\right)}{A} \right] \\
												 &+ \frac{1}{\ell} \left[\ell \left(\frac{n}{\omega_{u}-\omega_{l}}\right)\right] \left[ \frac{{\omega}_{u} -1}{A} \right] \\
												 &= \frac{1}{A} \left[\frac{n}{\omega_{u}-\omega_{l}}\right]  \left[\ln\left(\frac{1}{\omega_{l}}\right) + {\omega}_{u} - 1 \right] \ .
		\end{split}
	\end{equation}
	In general, this proportionality constant $A$ will depend on what
	method is used to measure correlation length.  For this reason,
	we look at the ratio ${{\chi}_{p}^{2}}/{{\chi}^{2}}$. Since
	both ${\chi}_{p}^{2}$ and ${\chi}^{2}$ have a ${1}/{A}$
	dependence, the ratio of the two is independent of $A$.
	
	We use this ratio ${\chi}_{p}^{2}/{\chi}^{2}$ to get an idea of the
	effect various upper and lower data cutoffs have on our fit.  
	We plot an example case of this analysis for systems of size $L=512$ 
	and $L=1024$ in dimension $d=2$.  In Figs.~\ref{fig:cutoffs}(a) and 
	\ref{fig:cutoffs}(b) it is apparent that below a certain $s_{l}$ threshold, 
	lattice effects skew the value of ${\chi}^{2}$ as well as the value of $d_{s}$. 
	Meanwhile, Figs.~\ref{fig:cutoffs}(c) and \ref{fig:cutoffs}(d)
	indicate that we need not impose an upper cutoff on the data.
	The high variance of points in the large $\omega$ region of the data 
	already appropriately weights these points.  For our final analysis,
	we chose $s_{l} = 425$ for two dimensions, $s_{l} = 375$ for three dimensions,
	and $s_{l} = 100$ for four and five dimensions. 
\begin{figure*} 
	\includegraphics[%
		width=\doublefigwidth, 
		keepaspectratio]{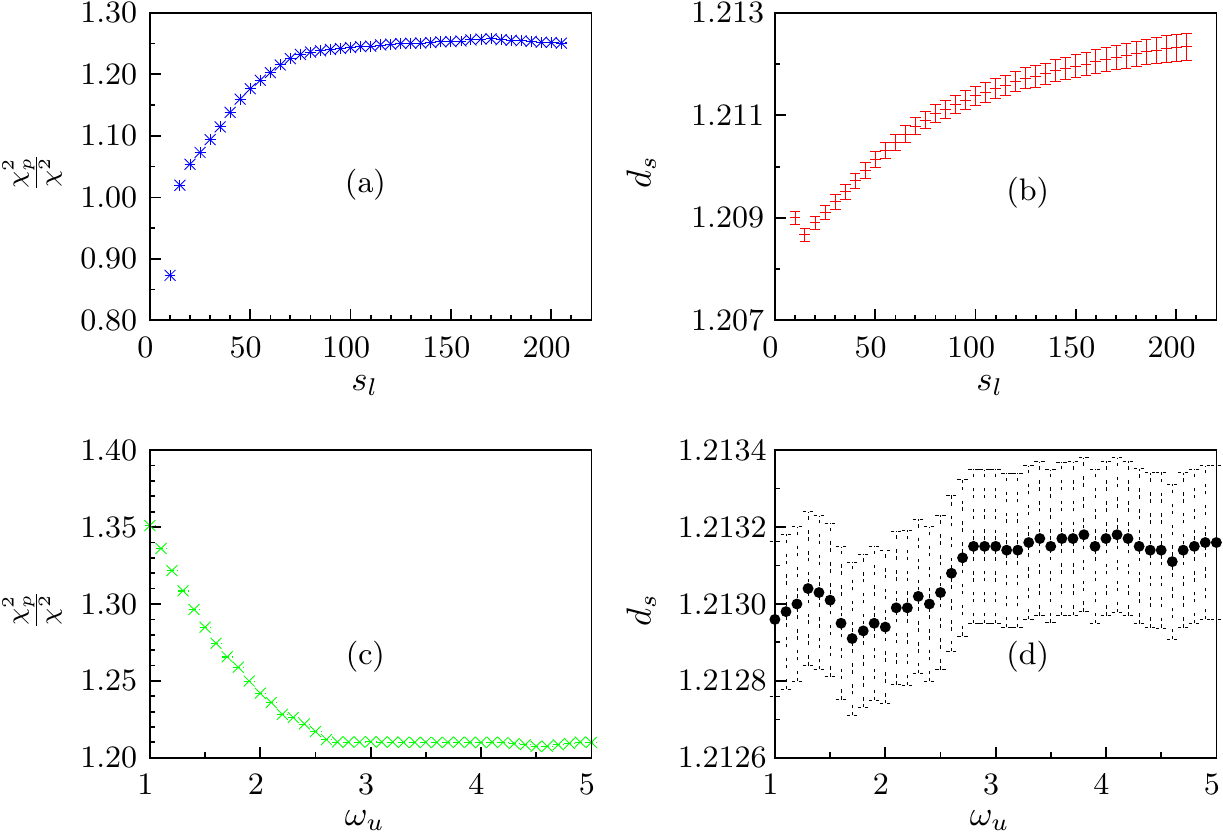}
	\caption{Examining the effects of the lower and upper data
		cutoffs on a collapse of two systems of sizes $L=512$ and $L=1024$ in 
		dimension $d=2$. (a) and (b) 
		display ${{\chi}_{p}^{2}}/{{\chi}^{2}}$ and $d_{s}$ as
		functions of the lower (small $s$) data cutoff $s_{l}$, 
		while (c) and (d) 
		show ${{\chi}_{p}^{2}}/{{\chi}^{2}}$ and $d_{s}$ as
		functions of the upper (large $s$) data cutoff ${\omega}_{u}$ for $s_{l} = 100$.
		Both ${\chi}_{p}^{2}$ and ${\chi}^{2}$ are measured at the value
		of $d_{s}$ for which ${\chi}^{2}$ is minimized.
	}
	\label{fig:cutoffs}
\end{figure*}

\bibliography{SM_biblio} 
\bibliographystyle{prsty}

\end{document}